\documentclass[preview]{elsarticle}
\usepackage[utf8]{inputenc}
\usepackage[english]{babel}
\usepackage{amsmath}
\usepackage{amsfonts}
\usepackage{amssymb}
\usepackage{amsthm}
\usepackage{graphicx}
\newtheorem{theorem}{Theorem}
\newtheorem{corollary}{Corollary}[theorem]

\usepackage{algorithm}
\usepackage{algpseudocode}
\usepackage[usenames,dvipsnames,svgnames,table]{xcolor}
\newcommand{\makematrixshortcut}[2]{%
	\expandafter\newcommand\csname #1n\endcsname{{#2_n}}%
	\expandafter\newcommand\csname #1ninv\endcsname{{#2_n}^{-1}}%
	\expandafter\newcommand\csname #1npinv\endcsname{{#2_n}^+}%
	\expandafter\newcommand\csname #1ntop\endcsname{{#2_n}^\top}%
	\expandafter\newcommand\csname #1nplusone\endcsname{#2_{n+1}}%
	\expandafter\newcommand\csname #1nplusoneinv\endcsname{{#2_{n+1}}^{-1}}%
	\expandafter\newcommand\csname #1nplusonepinv\endcsname{{#2_{n+1}}^+}%
	\expandafter\newcommand\csname #1nplusonetop\endcsname{{#2_{n+1}}^\top}%
	\expandafter\newcommand\csname #1ntilde\endcsname{\tilde{#2}_n}%
	\expandafter\newcommand\csname #1ntildeinv\endcsname{\tilde{#2_n}^{-1}}%
	\expandafter\newcommand\csname #1ntildetop\endcsname{\tilde{#2_n}^\top}%
	\expandafter\newcommand\csname #1nplusonetilde\endcsname{\tilde{#2}_{n+1}}%
	\expandafter\newcommand\csname #1nplusonetildeinv\endcsname{\tilde{#2_{n+1}}^{-1}}%
	\expandafter\newcommand\csname #1nplusonetildetop\endcsname{\tilde{#2_{n+1}}^\top}%
}
\newcommand{\makevectorshortcut}[2]{%
	\expandafter\newcommand\csname #1n\endcsname{{#2_n}}%
	\expandafter\newcommand\csname #1ntop\endcsname{{#2_n}^\top}%
	\expandafter\newcommand\csname #1nrej\endcsname{#2_n^r}%
	\expandafter\newcommand\csname #1nrejtop\endcsname{{#2_n^{rej}}^\top}%
	\expandafter\newcommand\csname #1nproj\endcsname{#2_n^p}%
	\expandafter\newcommand\csname #1nprojtop\endcsname{{#2_n^p}^\top}%
	\expandafter\newcommand\csname #1nplusone\endcsname{#2_{n+1}}%
	\expandafter\newcommand\csname #1nplusonetop\endcsname{{#2_{n+1}}^\top}%
	\expandafter\newcommand\csname #1nplusonerej\endcsname{#2_{n+1}^{rej}}%
	\expandafter\newcommand\csname #1nplusonerejtop\endcsname{{#2_{n+1}^{rej}}^\top}%
	\expandafter\newcommand\csname #1nplusoneproj\endcsname{#2_{n+1}^p}%
	\expandafter\newcommand\csname #1nplusoneprojtop\endcsname{{#2_{n+1}^p}^\top}%
}
\makematrixshortcut{A}{A}
\makematrixshortcut{B}{B}
\makematrixshortcut{C}{C}
\makematrixshortcut{P}{P}
\makevectorshortcut{G}{\Gamma}
\makevectorshortcut{g}{\gamma}
\makevectorshortcut{z}{\zeta}
\makevectorshortcut{b}{\beta}
\makevectorshortcut{x}{X}
\makevectorshortcut{y}{y}
\makevectorshortcut{Y}{Y}

\newcommand{\add}[1]{\textcolor{black}{ #1}}
\usepackage{hyperref}
\usepackage{numcompress}

\journal{Applied Mathematics and Computation}
\begin{document}
\begin{frontmatter}

\title{Efficient Recursive Least Squares Solver for Rank-Deficient Matrices}

\author{Ruben Staub\corref{mycorrespondingauthor}}
\ead{ruben.staub@ens-lyon.org}
\address{Univ Lyon, Ecole Normale Sup\'erieure de Lyon, CNRS Universit\'e Lyon
1, Laboratoire de Chimie  UMR 5182, 46 all\'ee d'Italie,  F-69364, LYON, France}
\author{Stephan N. Steinmann\corref{mycorrespondingauthor}}
\ead{stephan.steinmann@ens-lyon.fr}
\address{Univ Lyon, Ecole Normale Sup\'erieure de Lyon, CNRS Universit\'e Lyon
1, Laboratoire de Chimie  UMR 5182, 46 all\'ee d'Italie,  F-69364, LYON, France}

\begin{abstract}
Updating a linear least squares solution can be critical for near real-time signal-processing applications.
The Greville algorithm proposes a simple formula for updating the pseudoinverse of a matrix $A\in \mathbb{R}^{n\times m}$ with rank $r$. In this paper, we explicitly derive a similar formula by maintaining a general rank factorization, which we call rank-Greville.
Based on this formula, we implemented a recursive least squares algorithm exploiting the rank-deficiency of $A$, achieving the update of the minimum-norm least-squares solution  in $O(mr)$ operations and, therefore, solving the linear least-squares problem from scratch in $O(nmr)$ operations.
We empirically confirmed that this algorithm  displays a better asymptotic time complexity than LAPACK solvers for rank-deficient matrices. The numerical stability of rank-Greville was found to be comparable to Cholesky-based solvers. Nonetheless, our implementation supports exact numerical representations of rationals, due to its remarkable algebraic simplicity.
\end{abstract}


\end{frontmatter}


\section{Introduction}
In this paper, we are interested in computationally efficient algorithms for solving the recursive least-squares problem on rank-deficient matrices.

Let $\Gamma_i \in \mathbb{R}^m$ represent an observation of $m$ variables, called regressors, associated with measurement $y_i$ and let $\xn$ be the unknown parameters of the linear relation:
\begin{equation}\label{label_eq_linear_relation}
\An \xn + \epsilon_n = Y_n
\end{equation}
where $A_n$ is an $n\times m$ matrix representing $n$ such observations, $Y_n$ contains associated measurements and $\epsilon_n$ is the random disturbance:
\begin{equation}
\An = \begin{pmatrix}
{\Gamma_1}^\top \\
{\Gamma_2}^\top \\
\vdots \\
\Gntop
\end{pmatrix}, \quad Y_n = \begin{pmatrix}
y_1 \\
y_2 \\
\vdots \\
y_n
\end{pmatrix}
\end{equation}

The solution $\xn$ to the general linear least-squares problem of equation \ref{label_eq_linear_relation} is defined as:
\begin{equation}
\xn = \underset{x\in S}{\operatorname{arg\,min}}(||x||_2),\quad S = \underset{x}{\operatorname{arg\,min}}(||\An x - Y_n||_2)
\label{label_eq_least_squares}
\end{equation}
This solution is unique\cite{bjorck1996numerical}, and sometimes called minimum-norm least-squares solution. Because of its uniqueness, it is sometimes simply referred to as the least-squares solution.

As demonstrated in the seminal paper by Penrose, the least-squares solution $\xn$ can also be written\cite{penrose_1956_least_squares}:
\begin{equation}\label{label_eq_least_squares_pseudoinverse}
\xn = \Anpinv Y_n
\end{equation}
where $\Anpinv$ is the pseudoinverse of $\An$, also called generalized inverse, or Moore-Penrose inverse. Due to its practical importance, the numerical determination of the generalized inverse remains an active topic of research.\cite{rakha_moorepenrose_2004,toutounian_new_2009}

The pseudoinverse $A^+ \in \mathbb{C}^{m\times n}$ of any matrix $A \in \mathbb{C}^{n\times m}$ is uniquely\cite{moore1920reciprocalmatrix} characterized by the four Penrose equations\cite{penrose_1955_pseudoinverse}:
\begin{align}
A A^+ A &= A \label{label_eq_penrose:1}\\
A^+ A A^+ &= A^+ \label{label_eq_penrose:2}\\
(A A^+)^\top &= A A^+ \label{label_eq_penrose:3}\\
(A^+ A)^\top &= A^+ A \label{label_eq_penrose:4}
\end{align}


Here, we are interested in a particular problem, i.e., updating a least-squares solution (or the generalized inverse) when a new observation ($\Gnplusone$, $y_{n+1}$) is added. This is typically called the recursive least-squares (RLS) problem for which the updated solution $\xnplusone$ is usually written\cite{bjorck1996numerical, least_squares_overview}:
\begin{equation}\label{label_eq_recursive_least_squares}
\xnplusone = \xn + K \times (y_{n+1} - \Gnplusonetop\xn)
\end{equation}
where $y_{n+1} - \Gnplusonetop\xn$ is the predicted residual (or \textit{a priori} error), and $K$ is called the Kalman gain vector.

 Algorithms that allow to update an already known previous solution can be of critical importance for embedded-systems signal processing, for example, as near real-time solutions might be required and new observations are added continuously\cite{least_squares_overview}.
Therefore, recursive least squares algorithms significantly benefit from the computational efficiency introduced by updating the least square solution instead of recomputing it from scratch.


If $\An $ has full column rank, the recursive least-squares solution $\xnplusone$ of equation (\ref{label_eq_recursive_least_squares}) can be straightforwardly computed using normal equations\cite{bjorck1996numerical}. 
This RLS algorithm has a time complexity of $O(m^2)$ for each update. Therefore, computing the solution for $n$ successive observations, using equation (\ref{label_eq_recursive_least_squares}), lead to a total time complexity of $O(nm^2)$.

However in the general case, $\An$ can be rank deficient, i.e. neither full column rank nor full row rank. This is indeed the case if we want to sample a large variable space (column deficiency), while accumulating data redundancy on the subspace of observed variables (row deficiency). Handling rank deficiency is, for example, desirable in neurocomputational learning applications\cite{Courrieu2005FastCO}.
Several algorithms have been developed to solve the rank-defficient recursive least-squares problem. In particular, the Greville algorithm\cite{greville} was designed for the recursive least-squares problem specifically, whereas most of the other algorithms are common least-squares solvers (based on Cholesky decomposition, QR factorization, SVD, \ldots) adapted to support updates of $\An$ (without the need to recompute the whole solution)\cite{bjorck1996numerical}. 

The Greville algorithm provides \add{\iffalse{with}\fi} an updated least squares solution, and additionally, an updated pseudoinverse at the same cost. This update step still has computational complexity in $O(m^2)$, independently of the rank deficiency of $\An $. This leads to a $O(nm^2)$ time complexity for the computation of the full solution.\footnote{Note that one can easily reach $O(\max(n,m)\min(n,m)^2)$ using the property $(\Antop)^+ = (\Anpinv)^T$.} Variants of this algorithm were developed\cite{albert_sittler} based on the implicit decomposition of $\An$, but still with an $O(m^2)$ update complexity\footnote{Using the notations defined below, Albert and Sittler maintain the $m\times m$ matrices $(1 - \Cntop\Cntilde)$ and $(\Cntildetop\Pninv\Cntilde)$, whereas we maintain $\Cn$, $\Cntilde$ and $\Pninv$ (whose sizes are reduced).}.

In this paper, we write and implement a recursive least-squares algorithm that has single-update time complexity in $O(mr)$ (i.e. $O(nmr)$ total time complexity), where $r$ is the rank of the matrix $\An $. The underlying idea is to maintain a general rank decomposition into a full row rank matrix (purely underdetermined system) and a full column rank matrix (purely overdetermined system), which are much easier to treat in a recursive least squares procedure. Indeed, due to the rank deficiency of $\An $ these matrices have reduced sizes, leading to more efficient updates.

The remarkable simplicity of this approach makes it compatible with exact numerical representations in practice, without the need to use expensive symbolic computing. We also explore slightly more sophisticated rank decompositions, effectively bridging the gap between the Greville algorithm and QR-based recursive least-squares solvers.


\section{Derivation}
\subsection{General non-orthogonal rank factorization}

Let $\An $ be a $n\times m$ matrix of rank $r$. $\An $ can be expressed as the product\cite{greville}:
\begin{equation}\label{label_eq_mat_prod}
\An  = \Bn\Cn
\end{equation}
where $\Bn$ is a $n \times r$ matrix, and $\Cn$ is a $r \times m$ matrix, both of rank $r$, which corresponds to a full-rank factorization.

Let us consider a maximal free family $\mathcal{B_A}$ of observations among observations $\Gamma_{i} \in \mathbb{R}^m$ in $\An $. Note that $\mathcal{B_A}$ is a basis of $\operatorname{Im}(\Antop)$. Such basis is represented by the rows of $\Cn$. Interpreting rows of $\Cn$ as observations $\Gamma_{C_i} \in \mathbb{R}^m$, we find that each observation $\Gamma_{C_i}$ in $\Cn$ is linearly independent of the others. Hence $\Cn$ can then be seen as a purely underdetermined system. This system can be thought as linking the fitted value of each observation in $\mathcal{B_A}$, to the fitted values for the $m$ variables themselves.



Interpreting rows of $\Bn$ as observations $\gamma_{i} \in \mathbb{R}^r$, we find that each observation $\gamma_{i}$ in $\Bn$ is the observation $\Gamma_{i}$ of $\An$ expressed in the $\mathcal{B_A}$ basis. Therefore, $\Bn$ can be seen as a purely overdetermined system, since each observation in $\mathcal{B_A}$ is observed at least once. One can consider that this system links the value of each observation $\Gamma_{i}$ in $\An$, to the fitted values for each observation in $\mathcal{B_A}$.

\begin{theorem}\label{label_th_explicit_pseudo}
The pseudoinverse $\Anpinv$ of $\An$ can be computed in $O(nmr)$ if matrices $\Bn$ and $\Cn$ verifying equation \ref{label_eq_mat_prod} are known. An explicit formula is then given by:
\begin{equation}
\Anpinv = \Cntop(\Cn\Cntop)^{-1}(\Bntop\Bn)^{-1}\Bntop
\label{eq_explicit_pseudo}
\end{equation}
\end{theorem}
\begin{proof}
By definition of $\Bn$ and $\Cn$, $\Bntop\Bn$ and $\Cn\Cntop$ are both $r \times r$ matrices of rank $r$, and therefore non-singular. As a consequence, the explicit formula given is well defined as long as $\Bn$ and $\Cn$ correspond to a full-rank factorization.

It is straightforward to check that Eq. \ref{eq_explicit_pseudo} satisfies all four Penrose equations (Eq. \ref{label_eq_penrose:1} to \ref{label_eq_penrose:4}), and therefore, represents an acceptable pseudoinverse of $\An$\cite{greville}. By the unicity of the Moore-Penrose inverse, we conclude that $\Anpinv = \Cntop(\Cn\Cntop)^{-1}(\Bntop\Bn)^{-1}\Bntop$.
Computing the pseudoinverse can, therefore, be reduced to computing the inverse of two $r \times r$ matrices and three matrix multiplications giving rise (in any order) to a total $O(nmr)$ time complexity, that could even be reduced by using faster algorithms\cite{fast_rect_dot}.
\end{proof}

We are now interested in the update of the pseudoinverse of $\An$ when adding a new observation $\Gnplusone$. Let us define $\Anplusone$ as: 
\begin{equation}
	\Anplusone = \begin{pmatrix}
		\An \\
		\Gnplusonetop \\
	\end{pmatrix} = \Bnplusone \Cnplusone
\label{label_eq_new_rank_decompose}
\end{equation}

We distinguish two cases depending on the linear dependency of $\Gnplusone$ with respect to previous observations $\Gamma_1,\,\ldots,\,\Gamma_n$. Note that we can equally well only consider the observations in $\mathcal{B_A}$, since $\mathcal{B_A}$ is a basis for $\operatorname{Vect}(\Gamma_1,\,\ldots,\,\Gamma_n)$.

Let $P_\mathcal{B_A}$ be the projector into $\operatorname{Vect}(\mathcal{B_A}) = \operatorname{Im}(C^\top)$. We define $\Gnplusoneproj \in \mathbb{R}^m$ the projection of $\Gnplusone$ into $\operatorname{Vect}(\mathcal{B_A})$. We also define $\gnplusoneproj \in \mathbb{R}^r$ as $\Gnplusoneproj$ expressed in the $\mathcal{B_A}$ basis.
If $\mathcal{B_A}$ was an orthonormal basis, the decomposition $\gnplusoneproj$ could be easily computed by inner products with $\gnplusoneproj =\Cn\Gnplusone$. However, in the general (non-orthonormal) case, the decomposition $\gnplusoneproj$ of $P_\mathcal{B_A}\Gnplusone$ can be obtained using the dual $\tilde{\mathcal{B_A}}$ of $\mathcal{B_A}$, represented by the rows of $\Cntilde$ defined as:
\begin{equation}
\begin{aligned}
\Cntilde &= (\Cn\Cntop)^{-1}\Cn \\
 &= \sigma_n^{-1}\Cn \\
\end{aligned}
\end{equation}
where $\sigma_n = \Cn\Cntop$ is the Gram matrix of observations in $\mathcal{B_A}$. $\gnplusoneproj$ can then be expressed by:
\begin{equation}\label{label_eq_gnplusoneproj}
\begin{aligned}
	\gnplusoneproj &= \Cntilde P_\mathcal{B_A}\Gnplusone\\
	 &= \Cntilde\Gnplusone\\
\end{aligned}
\end{equation}

and $\Gnplusoneproj$ can then be expressed by:
\begin{equation}\label{label_eq_Gnplusoneproj}
\begin{aligned}
	\Gnplusoneproj &= P_\mathcal{B_A}\Gnplusone\\
	 &= \Cntop\Cntilde\Gnplusone\\
	 &= \Cntop\gnplusoneproj
\end{aligned}
\end{equation}

We define the rejection vector $\Gnplusonerej \in \mathbb{R}^m$ associated with the projection of $\Gamma_{n+1}$ into $\mathcal{B_A}$:
\begin{equation}
\begin{aligned}
\Gnplusone &= P_\mathcal{B_A}\Gnplusone + \Gnplusonerej\\
 &= \Gnplusoneproj + \Gnplusonerej\\
 &= \Cntop\gnplusoneproj + \Gnplusonerej\\
 &= \Cntop\Cntilde\Gnplusone + \Gnplusonerej
\end{aligned}
\end{equation}

It becomes clear that $\Gnplusone$ is linearly dependent from the previous observations $\Gamma_1,\,\ldots,\,\Gamma_n$, if and only if, $\Gnplusonerej$ is null. Note that $\gnplusoneproj$ , $\Gnplusoneproj$ and $\Gnplusonerej$ can be computed in $O(mr)$ if $\Cntilde$ and $\Cn$ are already known.

The pseudoinverse $\Anpinv$ can then be rewritten:
\begin{equation}
\Anpinv = \Cntildetop \Pninv \Bntop
\end{equation}
where 
\begin{equation}
\Pn = \Bntop \Bn = \sum\limits_{i=1}^n \gamma_i^p \cdot {\gamma_i^p}^\top
\end{equation}

We finally define $\znplusone \in \mathbb{R}^{r}$ and $\bnplusone \in \mathbb{R}^{n}$ as:
\begin{equation}\label{label_eq_znplusone}
\znplusone = \Pninv\gnplusoneproj,\quad \bnplusone = \Bn\znplusone
\end{equation}
Note that if $\Cntilde$, $\Cn$ and $\Pninv$ are already known, $\znplusone$ can be computed in $O(mr)$, but $\bnplusone$ can only be computed in $O(\max(m,n)r)$ if $\Bn$ is also known.

\begin{theorem}\label{label_th_pinv_indep}
If $\Gnplusone \neq 0$ is linearly independent from previous observations, the pseudoinverse $\Anplusonepinv$ of $\Anplusone$ can be updated in $O(mn)$ if $\Anpinv$, $\Bn$, $\Pninv$, $\Cn$ and $\Cntilde$ are known. An explicit formula is then given by:
\begin{equation}
\Anplusonepinv = \begin{pmatrix}
		\Anpinv & 0 \\
	\end{pmatrix} + \frac{\Gnplusonerej}{\left\|\Gnplusonerej\right\|_2^2} \begin{pmatrix}
		-\bnplusonetop & 1 \\
	\end{pmatrix}
\end{equation}
\end{theorem}
\begin{proof}
First, let us observe how the full-rank factorization $\Anplusone = \Bnplusone \Cnplusone$ is impacted by adding an observation $\Gnplusone$ linearly independent from previous observations. Concerning $\Cnplusone$, we have:
\begin{equation}\label{label_eq_im_Anplusone}
\begin{aligned}
	&\operatorname{Im}(\Cnplusonetop) = \operatorname{Im}(\Anplusonetop) = \operatorname{Vect}(\mathcal{B_A} \cup \{\Gnplusone\}) \neq \operatorname{Vect}(\mathcal{B_A}) = \operatorname{Im}(\Antop) = \operatorname{Im}(\Cntop)\\
	&\Rightarrow\quad \Cnplusone \neq \Cn
\end{aligned}
\end{equation}
Adding $\Gnplusone$ to the rows of $\Cn$ leads to:
\begin{equation}\label{label_eq_Cnplusone}
	\Cnplusone = \begin{pmatrix}
		\Cn \\
		\Gnplusonetop \\
	\end{pmatrix}, \quad \Bnplusone = \begin{pmatrix}
		\Bn & 0 \\
		0 & 1 \\
	\end{pmatrix}
\end{equation}
It becomes clear from this definition that $\Bnplusone$ has full column rank since $\Bn$ has full column rank, and also that $\Cnplusone$ has full row rank. Therefore, $\Bnplusone$ and $\Cnplusone$ represent an acceptable full-rank decomposition of $\Anplusone$, since we have:
\begin{equation}\label{label_eq_Anplusone_rank_decomposition_indep}
	\Bnplusone \Cnplusone = \begin{pmatrix}
		\Bn & 0 \\
		0 & 1 \\
	\end{pmatrix} \begin{pmatrix}
		\Cn \\
		\Gnplusonetop \\
	\end{pmatrix} = \begin{pmatrix}
		\Bn\Cn \\
		\Gnplusonetop \\
	\end{pmatrix} = \begin{pmatrix}
		\An \\
		\Gnplusonetop \\
	\end{pmatrix} = \Anplusone
\end{equation}
Second, we apply Theorem \ref{label_th_explicit_pseudo} to $\Anplusone$:
\begin{equation}\label{label_eq_Anplusonepinv_update_indep_proof}
\begin{aligned}
	\Anplusonepinv &= \Cnplusonetop(\Cnplusone\Cnplusonetop)^{-1}(\Bnplusonetop\Bnplusone)^{-1}\Bnplusonetop\\
	 &= \begin{pmatrix}
		\Cntop & \Gnplusone \\
	\end{pmatrix} \begin{pmatrix}
		\Cn\Cntop & \Cn\Gnplusone \\
		(\Cn\Gnplusone)^\top & \Gnplusonetop\Gnplusone \\
	\end{pmatrix}^{-1} \begin{pmatrix}
		\Bntop\Bn & 0 \\
		0 & 1 \\
	\end{pmatrix}^{-1} \begin{pmatrix}
		\Bntop & 0 \\
		0 & 1 \\
	\end{pmatrix}\\
	&= \begin{pmatrix}
		\Cntop & \Gnplusone \\
	\end{pmatrix} \begin{pmatrix}
		\Cn\Cntop & \Cn\Gnplusone \\
		(\Cn\Gnplusone)^\top & \Gnplusonetop\Gnplusone \\
	\end{pmatrix}^{-1} \begin{pmatrix}
		(\Bntop\Bn)^{-1} & 0 \\
		0 & 1 \\
	\end{pmatrix} \begin{pmatrix}
		\Bntop & 0 \\
		0 & 1 \\
	\end{pmatrix}\\
	&= \begin{pmatrix}
		\Cntop & \Gnplusone \\
	\end{pmatrix} \begin{pmatrix}
		\sigma_n & \Cn\Gnplusone \\
		(\Cn\Gnplusone)^\top & \Gnplusonetop\Gnplusone \\
	\end{pmatrix}^{-1} \begin{pmatrix}
		\Pninv\Bntop & 0 \\
		0 & 1 \\
	\end{pmatrix}
\end{aligned}
\end{equation}
Finally, we apply a generic block-wise inversion scheme:
\begin{equation}\label{label_eq_generic_block_inverse}
\begin{aligned}
\begin{pmatrix}
	D & E \\
	E^\top & F \\
\end{pmatrix}^{-1} &= \begin{pmatrix}
	D^{-1} + D^{-1}E(F-E^\top D^{-1}E)^{-1}E^\top D^{-1} & -D^{-1}E(F-E^\top D^{-1}E)^{-1} \\
	-(F-E^\top D^{-1}E)^{-1}E^\top D^{-1} & (F-E^\top D^{-1}E)^{-1} \\
\end{pmatrix} \\
	&= \begin{pmatrix}
		D^{-1} + D^{-1}ES^{-1}E^\top D^{-1} & -D^{-1}ES^{-1} \\
		-S^{-1}E^\top D^{-1} & S^{-1} \\
	\end{pmatrix}
\end{aligned}
\end{equation}
where $S = F-E^\top D^{-1}E$, $D = \sigma_n$, $E = \Cn\Gnplusone$ and $F = \Gnplusonetop\Gnplusone$, since $\sigma_n$ is non-singular and $\Gnplusonetop\Gnplusone$ is not null since $\Gnplusone \neq 0$.
This leads to the pseudoinverse formula:
\begin{equation}\label{label_eq_Anplusonepinv_update_indep_proof_final}
\begin{aligned}
	\Anplusonepinv &= \begin{pmatrix}
		\Cntop & \Gnplusone \\
	\end{pmatrix} \begin{pmatrix}
		\sigma_n & \Cn\Gnplusone \\
		(\Cn\Gnplusone)^\top & \Gnplusonetop\Gnplusone \\
	\end{pmatrix}^{-1} \begin{pmatrix}
		\Pninv\Bntop & 0 \\
		0 & 1 \\
	\end{pmatrix} \\
	&= \begin{pmatrix}
		\Cntop & \Gnplusone \\
	\end{pmatrix} \begin{pmatrix}
		\sigma_n^{-1} + \sigma_n^{-1}\Cn\Gnplusone S^{-1}(\Cn\Gnplusone)^\top\sigma_n^{-1} & -\sigma_n^{-1}\Cn\Gnplusone S^{-1} \\
		-S^{-1}(\Cn\Gnplusone)^\top\sigma_n^{-1} & S^{-1} \\
	\end{pmatrix} \begin{pmatrix}
		\Pninv\Bntop & 0 \\
		0 & 1 \\
	\end{pmatrix}\\
	&= \begin{pmatrix}
		\Cntop & \Gnplusone \\
	\end{pmatrix} \begin{pmatrix}
		\sigma_n^{-1} + \gnplusoneproj S^{-1}\gnplusoneprojtop & -S^{-1}\gnplusoneproj \\
		-S^{-1}\gnplusoneprojtop & S^{-1} \\
	\end{pmatrix} \begin{pmatrix}
		\Pninv\Bntop & 0 \\
		0 & 1 \\
	\end{pmatrix}\\
	&= \begin{pmatrix}
		\Cntop & \Gnplusone \\
	\end{pmatrix} \left( \begin{pmatrix}
		\sigma_n^{-1} & 0 \\
		0 & 0 \\
	\end{pmatrix} + S^{-1}\begin{pmatrix}
		\gnplusoneproj\gnplusoneprojtop & -\gnplusoneproj \\
		-\gnplusoneprojtop & 1 \\
	\end{pmatrix} \right) \begin{pmatrix}
		\Pninv\Bntop & 0 \\
		0 & 1 \\
	\end{pmatrix}\\
	&= \begin{pmatrix}
		\Anpinv & 0 \\
	\end{pmatrix} + S^{-1} \begin{pmatrix}
		\Cntop & \Gnplusone \\
	\end{pmatrix} \begin{pmatrix}
		\gnplusoneproj\gnplusoneprojtop & -\gnplusoneproj \\
		-\gnplusoneprojtop & 1 \\
	\end{pmatrix} \begin{pmatrix}
		\Pninv\Bntop & 0 \\
		0 & 1 \\
	\end{pmatrix}\\
	&= \begin{pmatrix}
		\Anpinv & 0 \\
	\end{pmatrix} + S^{-1} \begin{pmatrix}
		\Cntop\gnplusoneproj\gnplusoneprojtop\Pninv\Bntop - \Gnplusone\gnplusoneprojtop\Pninv\Bntop & - \Cntop\gnplusoneproj + \Gnplusone \\
	\end{pmatrix}\\
	&= \begin{pmatrix}
		\Anpinv & 0 \\
	\end{pmatrix} + S^{-1} \begin{pmatrix}
		-\Gnplusonerej\bnplusonetop & \Gnplusonerej \\
	\end{pmatrix}\\
\end{aligned}
\end{equation}
where the Schur complement $S$ of $\sigma_n$ is written as:
\begin{equation}\label{label_eq_schur_complement}
\begin{aligned}
S &= \Gnplusonetop\Gnplusone - \Gnplusonetop\Cntop\sigma_n^{-1}\Cn\Gnplusone\\
	&= \Gnplusonetop\Gnplusone - \Gnplusonetop\Gnplusoneproj\\
	&= \Gnplusonetop\Gnplusonerej\\
	&= \Gnplusonerejtop\Gnplusonerej = \left\|\Gnplusonerej\right\|_2^2
\end{aligned}
\end{equation}
$S$ is, therefore, the square of the norm of the component of $\Gnplusone$ along the orthogonal complement of $\operatorname{Im}(\Antop)$. $S$ is invertible since $\Gnplusonerej \neq 0$.

$\Gnplusonerej$ and $\bnplusone$ can be computed in $O(\max(m,n)r)$ if $\Cn$, $\Cntilde$, $\Bn$ and $\Pninv$ are already known. Therefore, the time complexity bottleneck is the outer product $\Gnplusonerej\bnplusonetop$, leading to a total update in $O(mn)$.
\end{proof}

\begin{theorem}\label{label_th_pinv_lindep}
If $\Gnplusone$ is a linear combination of previous observations, the pseudoinverse $\Anplusonepinv$ of $\Anplusone$ can be updated in $O(mn)$ if $\Anpinv$, $\Bn$, $\Pninv$, $\Cn$ and $\Cntilde$ are known. An explicit formula is then given by:
\begin{equation}
\Anplusonepinv = \begin{pmatrix}
		\Anpinv & 0 \\
	\end{pmatrix} + \frac{\Cntildetop \znplusone}{1+\gnplusoneprojtop\znplusone} \begin{pmatrix}
		-\bnplusonetop & 1\\
	\end{pmatrix}\\
\end{equation}
\end{theorem}
\begin{proof}
First, let us observe how the full-rank factorization $\Anplusone = \Bnplusone \Cnplusone$ is impacted by adding an observation $\Gnplusone$ that is a linear combination of previous observations. Notice that:
\begin{equation}\label{label_eq_Gnplusone_lin_dep}
\begin{aligned}
\Gnplusone &= P_\mathcal{B_A}\Gnplusone + \Gnplusonerej = P_\mathcal{B_A}\Gnplusone =  \Gnplusoneproj\\
 &= \Cntop\Cntilde\Gnplusone\\
 &= \Cntop\gnplusoneproj
\end{aligned}
\end{equation}
Since $\Gnplusone \in \operatorname{Vect}(\Gamma_1,\,\ldots,\,\Gamma_n) = \operatorname{Vect}(\mathcal{B_A})$, $\mathcal{B_A}$ is still a basis for $\operatorname{Vect}(\Gamma_1,\,\ldots,\,\Gn,\Gnplusone)$. As a consequence, we can take $\Cnplusone = \Cn$, leading to:
\begin{equation}\label{label_eq_factorization_lindep}
	\Cnplusone = \Cn, \quad \Bnplusone = \begin{pmatrix}
		\Bn\\
		\gnplusoneprojtop\\
	\end{pmatrix}
\end{equation}
From this definition follows that $\Cnplusone$ still has full row rank, and also that $\Bnplusone$ has full column rank since $\Bn$ has full column rank. Therefore, $\Bnplusone$ and $\Cnplusone$ represent an acceptable full-rank decomposition of $\Anplusone$, since we have:
\begin{equation}\label{label_eq_Anplusone_rank_decomposition_lindep}
	\Bnplusone \Cnplusone = \begin{pmatrix}
		\Bn\\
		\gnplusoneprojtop\\
	\end{pmatrix}\Cn = \begin{pmatrix}
		\Bn\Cn \\
		\gnplusoneprojtop\Cn \\
	\end{pmatrix} = \begin{pmatrix}
		\An \\
		\Gnplusonetop \\
	\end{pmatrix} = \Anplusone
\end{equation}
Second, we apply Theorem \ref{label_th_explicit_pseudo} to $\Anplusone$:
\begin{equation}\label{label_eq_Anplusonepinv_update_lindep_proof}
\begin{aligned}
	\Anplusonepinv &= \Cnplusonetop(\Cnplusone\Cnplusonetop)^{-1}(\Bnplusonetop\Bnplusone)^{-1}\Bnplusonetop\\
	&= \Cntop(\Cn\Cntop)^{-1} \left(\begin{pmatrix}
			\Bntop & \gnplusoneproj\\
		\end{pmatrix}\begin{pmatrix}
			\Bn\\
			\gnplusoneprojtop\\
		\end{pmatrix}\right)^{-1} \begin{pmatrix}
		\Bntop & \gnplusoneproj\\
	\end{pmatrix}\\
	&= \Cntop(\Cn\Cntop)^{-1} \left(\Bntop\Bn + \gnplusoneproj\gnplusoneprojtop\right)^{-1} \begin{pmatrix}
		\Bntop & \gnplusoneproj\\
	\end{pmatrix}\\
	&= \Cntildetop \left(\Pn + \gnplusoneproj\gnplusoneprojtop\right)^{-1} \begin{pmatrix}
		\Bntop & \gnplusoneproj\\
	\end{pmatrix}
\end{aligned}
\end{equation}
Finally, we apply the Sherman-Morrison formula stating that for any non-singular matrix $G \in \mathbb{R}^{n\times n}$ and any vector $v \in \mathbb{R}^{n}$, if $G + vv^\top$ is non-singular, then:
\begin{equation}\label{label_eq_sherman_morrison}
\add{\left(G + vv^\top\right)^{-1}} = G^{-1} - \frac{G^{-1}vv^\top G^{-1}}{1+v^\top G^{-1}v}
\end{equation}
with $G = \Pn$ and $v = \gnplusoneproj$, since $\Pn$ and $\Pnplusone = \Bnplusone\Bnplusonetop$ are non-singular\footnote{$\Pnplusone = \Bnplusone\Bnplusonetop$ is non-singular, since $\Pnplusone$ is square with full rank. Indeed, $\operatorname{rank}(\Pnplusone) = \operatorname{rank}(\Bnplusone) = r$, using theorem 5.5.4 of \cite{mirsky1990introduction}}.
This leads to the pseudoinverse formula:
\begin{equation}\label{label_eq_Anplusonepinv_update_lindep_proof_final}
\begin{aligned}
	\Anplusonepinv &= \Cntildetop \left(\Pn + \gnplusoneproj\gnplusoneprojtop\right)^{-1} \begin{pmatrix}
		\Bntop & \gnplusoneproj\\
	\end{pmatrix}\\
	&= \Cntildetop \left(\Pninv - \frac{\Pninv\gnplusoneproj\gnplusoneprojtop\Pninv}{1+\gnplusoneprojtop\Pninv\gnplusoneproj}\right) \begin{pmatrix}
		\Bntop & \gnplusoneproj\\
	\end{pmatrix}\\
	&= \Cntildetop \left(\Pninv - \frac{\znplusone\znplusonetop}{1+\gnplusoneprojtop\znplusone}\right) \begin{pmatrix}
		\Bntop & \gnplusoneproj\\
	\end{pmatrix}\\
	&= \begin{pmatrix}
		\Anpinv & \Cntildetop\znplusone \\
	\end{pmatrix} - \Cntildetop \frac{\znplusone\znplusonetop}{1+\gnplusoneprojtop\znplusone} \begin{pmatrix}
		\Bntop & \gnplusoneproj\\
	\end{pmatrix}\\
	&= \begin{pmatrix}
		\Anpinv & \Cntildetop\znplusone \\
	\end{pmatrix} - \frac{\Cntildetop \znplusone}{1+\gnplusoneprojtop\znplusone} \begin{pmatrix}
		\bnplusonetop & \znplusonetop\gnplusoneproj\\
	\end{pmatrix}\\
	&= \begin{pmatrix}
		\Anpinv & 0 \\
	\end{pmatrix} + \frac{\Cntildetop \znplusone}{1+\gnplusoneprojtop\znplusone} \begin{pmatrix}
		-\bnplusonetop & 1 +\gnplusoneprojtop\znplusone - \znplusonetop\gnplusoneproj\\
	\end{pmatrix}\\
	&= \begin{pmatrix}
		\Anpinv & 0 \\
	\end{pmatrix} + \frac{\Cntildetop \znplusone}{1+\gnplusoneprojtop\znplusone} \begin{pmatrix}
		-\bnplusonetop & 1\\
	\end{pmatrix}\\
\end{aligned}
\end{equation}
$\gnplusoneproj$ , $\znplusone$ and $\bnplusone$ can be computed in $O(\max(m,n)r)$ if $\Cntilde$, $\Bn$ and $\Pninv$ are already known. Therefore, the time complexity bottleneck is the outer product $\left(\Cntildetop \znplusone\right)\bnplusonetop$, leading to a total update in $O(mn)$.
\end{proof}

\begin{corollary}
For any observation $\Gnplusone \in \mathbb{R}^m$, the pseudoinverse $\Anplusonepinv$ of $\Anplusone$ can be updated in $\Theta(mn)$ if $\Anpinv$, $\Bn$, $\Pninv$, $\Cn$ and $\Cntilde$ are known.
\end{corollary}

Indeed, at least $n\times m$ terms of the pseudoinverse need to be updated when adding a new observation, in the general case\footnote{To be convinced, consider $\An = I_n$ the identity matrix and $\Gnplusone = \begin{pmatrix}
	1\\
	\vdots\\
	1\\
\end{pmatrix}$.}. Therefore, the pseudoinverse update has a fundamental cost component that cannot be improved, hence the $\Theta(mn)$ complexity. This limitation is not present in the recursive least square problem. In this problem, we are only interested in updating the least square solution $\xnplusone$ when adding a new observation $\Gnplusone$ with associated target $\ynplusone$:
\begin{equation}\label{label_eq_RLS_xnplusone}
\xnplusone = \begin{pmatrix}
\An\\
\Gnplusonetop\\
\end{pmatrix}^{+} \begin{pmatrix}
	\Yn\\
	\ynplusone\\
\end{pmatrix} = \Anplusonepinv \begin{pmatrix}
	\Yn\\
	\ynplusone\\
\end{pmatrix}
\end{equation}

\begin{theorem}\label{label_th_xnplusone_indep}
If $\Gnplusone \neq 0$ is linearly independent from previous observations, the least square solution $\xnplusone$ can be updated in $O(mr)$ if $\xn$, $\Cn$ and $\Cntilde$ are known. An explicit formula (in the form of equation \ref{label_eq_recursive_least_squares}) is then given by:
\begin{equation}
\xnplusone = \xn + \frac{\Gnplusonerej}{\left\|\Gnplusonerej\right\|_2^2} \left(\ynplusone - \Gnplusonetop\xn\right)
\end{equation}
\end{theorem}
\begin{proof}
First, let us inject theorem \ref{label_th_pinv_indep} into the definition of $\xnplusone$:
\begin{equation}\label{label_eq_xnplusone_indep}
\begin{aligned}
	\xnplusone &= \Anplusonepinv \begin{pmatrix}
		\Yn\\
		\ynplusone\\
	\end{pmatrix}\\
	&= \Anpinv \Yn + \frac{\Gnplusonerej}{\left\|\Gnplusonerej\right\|_2^2} \begin{pmatrix}
		-\bnplusonetop & 1 \\
	\end{pmatrix} \begin{pmatrix}
		\Yn\\
		\ynplusone\\
	\end{pmatrix}\\
	&= \xn + \frac{\Gnplusonerej}{\left\|\Gnplusonerej\right\|_2^2}\left(\ynplusone -\bnplusonetop \Yn\right)\\
\end{aligned}
\end{equation}
Let us simplify further this equation by recognizing $\bnplusonetop \Yn$ as the fitted target associated with $\Gnplusone$:
\begin{equation}\label{label_eq_xnplusone_indep_final}
\begin{aligned}
	\xnplusone &= \xn + \frac{\Gnplusonerej}{\left\|\Gnplusonerej\right\|_2^2} \left(\ynplusone -\Gnplusonetop\Cntildetop \Pninv \Bntop \Yn\right)\\
	&= \xn + \frac{\Gnplusonerej}{\left\|\Gnplusonerej\right\|_2^2} \left(\ynplusone -\Gnplusonetop\Anpinv \Yn\right)\\
	&= \xn + \frac{\Gnplusonerej}{\left\|\Gnplusonerej\right\|_2^2} \left(\ynplusone -\Gnplusonetop\xn\right)\\
	&= \xn + \frac{\Gnplusonerej}{\left\|\Gnplusonerej\right\|_2^2}\ \Delta\ynplusone\\
\end{aligned}
\end{equation}
where $\Delta\ynplusone = \ynplusone - \Gnplusonetop\xn$ is the difference between the expected/fitted target (i.e. $\Gnplusone\xntop$) and the real target $\ynplusone$ associated with the new observation $\Gnplusone$ (i.e. the predicted residual, or \textit{a priori} error). We identify $\frac{\Gnplusonerej}{\left\|\Gnplusonerej\right\|_2^2}$ to be the associated Kalman gain vector in this case\cite{bjorck1996numerical}.

$\Gnplusonerej$ can be computed in $O(mr)$ if $\Cn$ and $\Cntilde$ are already known, which is the time complexity bottleneck of the whole update step.
\end{proof}

\begin{theorem}\label{label_th_xnplusone_lindep}
If $\Gnplusone$ is a linear combination of previous observations, the least square solution $\xnplusone$ can be updated in $O(mr)$ if $\xn$, $\Cn$, $\Cntilde$ and $\Pninv$ are known. An explicit formula (in the form of equation \ref{label_eq_recursive_least_squares}) is then given by:
\begin{equation}
\xnplusone = \xn + \frac{\Cntildetop \znplusone}{1+\gnplusoneprojtop\znplusone} \left(\ynplusone - \Gnplusonetop\xn\right)\\
\end{equation}
\end{theorem}
\begin{proof}
Let us proceed similarly to theorem \ref{label_th_xnplusone_indep}, by injecting theorem \ref{label_th_pinv_lindep} into the definition of $\xnplusone$:
\begin{equation}\label{label_eq_xnplusone_lindep}
\begin{aligned}
	\xnplusone &= \Anplusonepinv \begin{pmatrix}
		\Yn\\
		\ynplusone\\
	\end{pmatrix}\\
	&= \Anpinv \Yn + \frac{\Cntildetop \znplusone}{1+\gnplusoneprojtop\znplusone} \begin{pmatrix}
		-\bnplusonetop & 1\\
	\end{pmatrix} \begin{pmatrix}
		\Yn\\
		\ynplusone\\
	\end{pmatrix}\\
	&= \xn + \frac{\Cntildetop \znplusone}{1+\gnplusoneprojtop\znplusone} \left(\ynplusone -\bnplusonetop \Yn\right)\\
	&= \xn + \frac{\Cntildetop \znplusone}{1+\gnplusoneprojtop\znplusone} \left(\ynplusone - \Gnplusonetop\xn\right)\\
	&= \xn + \frac{\Cntildetop \znplusone}{1+\gnplusoneprojtop\znplusone} \Delta\ynplusone\\
\end{aligned}
\end{equation}
where $\Delta\ynplusone = \ynplusone - \Gnplusone\xntop$ (i.e. the predicted residual, or \textit{a priori} error). We identify $\frac{\Cntildetop \znplusone}{1+\gnplusoneprojtop\znplusone}$ to be the associated Kalman gain vector in this case.

$\gnplusoneproj$ and $\znplusone$ can be computed in $O(mr)$ if $\Cn$, $\Cntilde$ and $\Pninv$ are already known. The whole update step can then be performed in $O(mr)$ operations.
\end{proof}

\begin{theorem}\label{label_th_update_decomp}
For any new observation $\Gnplusone \in \mathbb{R}^m$, the matrices $\Cnplusone$, $\Cnplusonetilde$ and $\Pnplusoneinv$ can be updated in $O(mr)$ if $\Cn$, $\Cntilde$ and $\Pninv$ are already known.
\end{theorem}
\begin{proof}
The updating formula naturally depends on the linear dependency of $\Gnplusone$ from previous observations (i.e. whether $\Gnplusonerej$ is non-null). Let us note that $\Gnplusonerej$ itself can be computed in $O(mr)$ operations if   $\Cntilde$ and $\Cn$ are known.

If $\Gnplusone$ is linearly independent from previous observations (i.e. $\Gnplusonerej \neq 0$), equation \ref{label_eq_Cnplusone} is valid, leading to:
\begin{equation}\label{label_eq_Pnplusoneinv_indep_update}
\begin{aligned}
\Pnplusone^{-1} &= \left(\Bnplusonetop \Bnplusone\right)^{-1} = \left(\begin{pmatrix}
		\Bntop & 0 \\
		0 & 1 \\
	\end{pmatrix} \begin{pmatrix}
		\Bn & 0 \\
		0 & 1 \\
	\end{pmatrix}\right)^{-1} = \begin{pmatrix}
		\Bntop\Bn & 0 \\
		0 & 1 \\
	\end{pmatrix}^{-1}\\
	&= \begin{pmatrix}
		\Pninv & 0 \\
		0 & 1 \\
	\end{pmatrix}\\
\end{aligned}
\end{equation}
\begin{equation}\label{label_eq_Cnplusone_indep_update}
\begin{aligned}
\Cnplusone = \begin{pmatrix}
		\Cn \\
		\Gnplusonetop \\
	\end{pmatrix}
\end{aligned}
\end{equation}
Using equation \ref{label_eq_Anplusonepinv_update_indep_proof_final}, we can write:
\begin{equation}\label{label_eq_Cnplusonetilde_indep_update}
\begin{aligned}
\Cnplusonetilde &= (\Cnplusone\Cnplusonetop)^{-1}\Cnplusone
 = \begin{pmatrix}
		\sigma_n^{-1} + \frac{\gnplusoneproj \gnplusoneprojtop}{\left\|\Gnplusonerej\right\|_2^2} & -\frac{\gnplusoneproj}{\left\|\Gnplusonerej\right\|_2^2} \\
		-\frac{\gnplusoneprojtop}{\left\|\Gnplusonerej\right\|_2^2} & \frac{1}{\left\|\Gnplusonerej\right\|_2^2} \\
	\end{pmatrix} \begin{pmatrix}
		\Cn \\
		\Gnplusone \\
	\end{pmatrix}\\
    &= \begin{pmatrix}
		\Cntilde - \gnplusoneproj\frac{\Gnplusonerejtop}{\left\|\Gnplusonerej\right\|_2^2} \\
		\frac{\Gnplusonerejtop}{\left\|\Gnplusonerej\right\|_2^2} \\
	\end{pmatrix}
\end{aligned}
\end{equation}
These formulae can be applied in $O(mr)$ operations, since $\gnplusoneproj$ and $\Gnplusonerej$ can themselves be computed in $O(mr)$ operations if $\Cntilde$ and $\Cn$ are already known.

If $\Gnplusone$ is a linear combination of previous observations (i.e. $\Gnplusonerej = 0$), equation \ref{label_eq_factorization_lindep} is valid, leading to:
\begin{equation}\label{label_eq_Cnplusone_Cnplusonetilde_lindep_update}
\begin{aligned}
\Cnplusone = \Cn,\quad \Cnplusonetilde = \Cntilde
\end{aligned}
\end{equation}
Using equation \ref{label_eq_Anplusonepinv_update_lindep_proof_final}, we can write:
\begin{equation}\label{label_eq_Pnplusoneinv_lindep_update}
\begin{aligned}
\Pnplusoneinv &= \Pninv - \frac{\znplusone\znplusonetop}{1+\gnplusoneprojtop\znplusone}
\end{aligned}
\end{equation}
This formula can be applied in $O(mr)$ operations\footnote{Note that this complexity reduces to $O(r^2)$ if $\gnplusoneproj$ is already known.}, since $\gnplusoneproj$ and $\znplusone$ can themselves be computed in $O(mr)$ operations if $\Pninv$, $\Cntilde$ and $\Cn$ are already known.

\end{proof}

\begin{corollary}\label{label_th_full_lstsq}
For any $n\times m$ matrix $\An$ of rank $r$ and any vector $\Yn \in \mathbb{R}^n$, the least square solution $\xn = \Anpinv \Yn$ can be computed in $O(mnr)$ operations.
\end{corollary}

\subsection{Orthogonal rank factorization}
\label{label_orthog_variant}

The theorems regarding the update of the pseudo-inverse $\Anplusonepinv$ (theorems \ref{label_th_pinv_indep} and \ref{label_th_pinv_lindep}) and least-squares solution $\xnplusone$ (theorems \ref{label_th_xnplusone_indep} and \ref{label_th_xnplusone_lindep}) are valid for any decomposition satisfying equation \ref{label_eq_mat_prod}.
Therefore, the rows of $\Cn$ can be required to form an orthogonal basis. This can be easily achieved by storing only the rejection vectors $\Gnplusonerej$ into $\Cnplusone$. More generally, one can store rescaled rejection vectors $\alpha_{n+1}\Gnplusonerej$ instead, with $0\neq\alpha_{n+1}\in\mathbb{R}$. Theorem \ref{label_th_update_decomp} remains valid, with equations \ref{label_eq_Pnplusoneinv_indep_update}, \ref{label_eq_Cnplusone_indep_update} and \ref{label_eq_Cnplusonetilde_indep_update} becoming:
\begin{equation}\label{label_eq_Cnplusone_indep_update_orthog}
\begin{aligned}
\Cnplusone = \begin{pmatrix}
		\Cn \\
		\alpha_{n+1}\Gnplusonerejtop \\
	\end{pmatrix}
\end{aligned}
\end{equation}
\begin{equation}\label{label_eq_Cnplusonetilde_indep_update_orthog}
\begin{aligned}
\Cnplusonetilde = \begin{pmatrix}
		\Cntilde \\
		\frac{\Gnplusonerejtop}{\alpha_{n+1}\left\|\Gnplusonerej\right\|_2^2} \\
	\end{pmatrix}
\end{aligned}
\end{equation}
\begin{equation}\label{label_eq_Pnplusoneinv_indep_update_orthog}
\begin{aligned}
\Pnplusone^{-1} &= \left(\Bnplusonetop \Bnplusone\right)^{-1} \\
                &= \left(\begin{pmatrix}
		\Bntop & \gnplusoneproj \\
		0 & \frac{1}{\alpha_{n+1}} \\
	\end{pmatrix} \begin{pmatrix}
		\Bn & 0 \\
		\gnplusoneprojtop & \frac{1}{\alpha_{n+1}} \\
	\end{pmatrix}\right)^{-1} \\
                &= \begin{pmatrix}
		\Pninv & -\alpha_{n+1}\znplusone \\
		-\alpha_{n+1}\znplusonetop & \alpha_{n+1}^2(1+\gnplusoneprojtop\znplusone) \\
	\end{pmatrix}
\end{aligned}
\end{equation}

In particular, one can consider $\alpha_{n+1} = \left\|\Gnplusonerej\right\|_2^{-1}$. In this case the rows of $\Cn$ form an orthonormal basis, i.e., $\Cn$ is an orthogonal matrix (i.e. $\Cntilde = \Cn$), and equation \ref{label_eq_mat_prod} becomes a Gram-Schmidt based thin LQ decomposition. This variant offers slightly reduced storage and update computational time.

\section{Implementation}
Based on the theorems above, one can devise a simple algorithm satisfying corollary \ref{label_th_full_lstsq}, the pseudocode of which is shown in Algorithm \ref{label_algo_RLS}.
\begin{algorithm}
\caption{\label{label_algo_RLS}Rank-deficient RLS, also called rank-Greville}
\begin{algorithmic}[1]
\Procedure{UpdateLeastSquares}{$\Gamma$, $y$, $X$, $C$, $\tilde{C}$, $P^{-1}$}
\State $\gamma \gets \tilde{C}\,\Gamma$
\State $\Gamma_r \gets \Gamma - C^\top\gamma$
\If {$\Gamma_r \neq 0$}
\State $K \gets \frac{\Gamma_r}{\left\|\Gamma_r\right\|_2^2}$
\State $C \gets \begin{pmatrix}
		C \\
		\Gamma^\top \\
	\end{pmatrix}$
\State $\tilde{C} \gets \begin{pmatrix}
		\tilde{C} - \gamma\,K^\top \\
		K^\top \\
	\end{pmatrix}$
\State $P^{-1} \gets \begin{pmatrix}
		P^{-1} & 0 \\
		0 & 1 \\
	\end{pmatrix}$
\Else
\State $\zeta \gets P^{-1}\gamma$
\State $K \gets \frac{\tilde{C}^\top\zeta}{1+\gamma^\top\zeta}$
\State $P^{-1} \gets P^{-1} - \frac{\zeta\,\zeta^\top}{1+\gamma^\top\zeta}$
\EndIf
\State $X \gets X + K\times\left(y - \Gamma^\top X\right)$
\EndProcedure
\end{algorithmic}
\end{algorithm}

These formulae have been implemented in Python3 using the Numpy library. In addition to the least-squares update algorithm (in $O(mr)$ operations), this implementation also supports pseudo-inverse (in $O(mn)$ operations) and covariance matrix\cite{greene2003econometric} updates (in $O(m^2)$ operations).

The orthogonal and orthonormal basis variants described in section \ref{label_orthog_variant} have also been implemented for least-squares update (in $O(mr)$ operations) and pseudo-inverse update (in $O(mn)$ operations).

In practice, checking if $\Gnplusonerej \neq 0$ is ill-defined with floating-point arithmetic. Yet, it is crucial in this algorithm as it determines the effective rank of the linear system.{\iffalse{Linear solvers implementations commonly determine the effective rank though a $rcond$ parameter. This parameter defines the minimal estimated reciprocal condition number acceptable for $A$.\label{label_def_rcond}}\fi} Therefore, we define a threshold $eps$ so that $\Gnplusone$ is considered a linear combination of previous observations if and only if:
\begin{equation}\label{label_eq_eps}
||\Gnplusonerej||_2 < eps \quad\text{or}\quad ||\Gnplusonerej||_2 < eps\times ||\Gnplusone||_2
\end{equation}
By default, \textit{eps} is set to $(m^2r+mr+m)\times\epsilon_M$ in order to account for rounding error propagation, with $\epsilon_M$ being the machine precision.

It is important to note that the general non-orthogonal basis implementation does support exact representations such as those defined in Python's \textit{fractions} module. Indeed, this scheme (as the Greville algorithm) uses only operations well defined on rational numbers.{\iffalse{ This property, its generality and remarkable simplicity were the reason we focused on its implementation.}\fi} One should note that the orthogonal basis variant scheme is also compatible with exact numerical representations as long as the rescaling factors $\alpha_{n+1}$ can themselves be represented exactly.
\section{Numerical Tests}

All computations were performed using Python 3.6.9\cite{Python3,IPython} with Scipy version 1.4.1\cite{2020SciPy-NMeth} and Numpy version 1.18.3\cite{Numpy, Numpy2} linked with OpenBLAS on an Intel Xeon W-2123 CPU with DDR4-2666 RAM. The code used for the numerical tests is available along with our Python3 implementation in the supporting information and will soon become available on \url{https://github.com/RubenStaub/rank-greville}.

\subsection{Computational efficiency}

In this section we empirically evaluate the computational efficiency achieved by the rank factorization Greville algorithm described in this paper (Algorithm \ref{label_algo_RLS}), referred to as "rank-Greville". Note that for comparison purposes, its total complexity (i.e. computing the full least-squares solution from scratch in $O(nmr)$ operations) is studied here, even though this algorithm was not specifically designed as a least-squares solver from scratch\footnote{Indeed, in this case rank-Greville must perform much more memory allocation and copy than a conventional solver, explaining, in part, the large prefactor}. Rather, rank-Greville was designed for the recursive update in applications requiring a constantly up-to-date solution.

For this analysis to be meaningful, the rank-Greville algorithm time efficiency is compared against standard algorithms from the LAPACK library:
\label{label_LAPACK_desc}
\begin{itemize}
\item "gelsy" refers to the DGELSY least-squares solver from the LAPACK library based on a complete orthogonal factorization (using QR factorization with column pivoting)\cite{lapack_guide, higham_accuracy_algo}.
\item "gelss" refers to the DGELSS least-squares solver from the LAPACK library based on SVD\cite{lapack_guide}.
\item "gelsd" refers to the DGELSD least-squares solver from the LAPACK library also based on SVD, using a diagonal form reduction\cite{lapack_guide}.
\end{itemize}
These routines were accessed through the \textit{scipy.linalg.lstsq} wrapper function from the Scipy library.

For these tests, we consider the computation from scratch of the least-squares solution $\xn$ verifying:
\begin{equation}\label{label_eq_tests_speed}
R_\mathcal{N}(n,m,r)\xn + \epsilon_n = R_\mathcal{N}(n,1,1)
\end{equation} where \add{$R_\mathcal{N}(n,m,r) \in \mathbb{R}^{n\times m}$ is a pseudo-random matrix with rank $r$ and whose elements are standardized ($\mu=0$, $\sigma=1$)\footnote{While random matrices are generally non rank-deficient, rank-deficient pseudo-random matrices are obtained here from the dot product of adequate pseudo-random matrices generated using the normal $\mathcal{N}(0,1)$ distribution}.} For reproducibility purposes, the pseudo-random number generator was reset before each $R_\mathcal{N}(n,m,r)$ computation.

In order to assess the scaling properties of these algorithms, we evaluate (using Python's \textit{timeit} module) the time elapsed for solving equation \ref{label_eq_tests_speed} \add{using various input sizes.
In the following, it is assumed that sufficiently large matrices were used for each algorithm to approach its asymptotic behavior in terms of computational time dependency with respect to input size. In other words, we assume that the scaling properties observed reflect the computational complexity of the algorithms considered\footnote{Which is equivalent to assuming that, for each algorithm, the largest computational times are dominated by the highest order term for the variable being considered.}.} 
These analyses were performed using various ranges for the parameters $n$, $m$, and $r$, \add{allowing for an empirical characterization of the computational complexities involved}:
\begin{itemize}
\item Full-rank square matrices:
\begin{equation}\label{label_eq_test_square}
r = n = m
\end{equation} 
Figure \ref{label_fig_cost_square} highlights a \add{roughly $n^3$ scaling of the elapsed computational time for large enough matrices, for all algorithms.}

\add{This corroborates the} $O(n^3)$ asymptotic complexity of all algorithms in this case.
\item Full-rank rectangular matrices:
\begin{equation}\label{label_eq_test_rect}
r = n \leq m
\end{equation} 
When fixing the number of rows $n$, Figure \ref{label_fig_cost_rect} highlights a \add{roughly linear scaling (with respect to $m$) of the elapsed computational time for large enough matrices, for all algorithms}\footnote{Additional tests at $m > 4\times 10^5$ seem to confirm the asymptotic linear dependency on $m$ for the DGELSS solver.}. 

\add{Since square matrices are a special case of rectangular matrices, this result is expected to be compatible with the scaling properties found for square matrices, if we assume a similar processing of both inputs. Therefore, in the case of full-rank rectangular matrices (with $n \leq m$), only an empirical $mn^2$ scaling (with respect to $n$ and $m$) can be deduced from these experiments.}

\add{This supports a} $O(mn^2)$ asymptotic complexity of all algorithms in this case.
\item Rank-deficient square matrices:
\begin{equation}\label{label_eq_test_rank}
r \leq n = m
\end{equation} 
At fixed rank $r=100$, \add{with square matrices of increasing size,} Figure \ref{label_fig_cost_rank} nicely highlights the \add{{\iffalse{specific}\fi}particular scaling (with respect to $n$) of rank-Greville, compared to the other commonly used algorithms.
Indeed, all LAPACK-based algorithms display a roughly $n^3$ empirical scaling (comparable with results found for full-rank square matrices), while rank-Greville was found to display a roughly $n^2$ scaling of the elapsed computational time for large enough matrices.}

\add{Using a similar argument as before, and assuming rank-deficient matrices are not treated differently, these results are expected to be compatible with the previously found scaling properties. Therefore, the only empirical scaling (with respect to $r$, $n$ and $m$) compatible with all previous experiments is in $mnr$ for rank-Greville and $mn^2$ for all other algorithms.}

\add{These empirical findings {\iffalse{substantiate}\fi}demonstrate the distinctive ability of rank-Greville to take advantage of rank-deficiency, in order to} reach an $O(mnr)$ asymptotic complexity.
\item Optimal (time-efficiency wise) applicability domains: 
\begin{equation}\label{label_eq_test_ratio}
r \leq n \leq m
\end{equation}
Figure \ref{label_fig_domain_ratio} illustrates which algorithm is the fastest (and by which margin) for a range of parameters ratios. Even though not specifically designed for solving the linear least-squares problem from scratch, the rank-Greville algorithm appears to be more efficient than other LAPACK solvers, but only when the observations matrix has particularly low rank $r \lesssim 0.15\times\min(n,m)$.
\end{itemize}

\begin{figure}[h!]
\centering
\includegraphics[width=\linewidth]{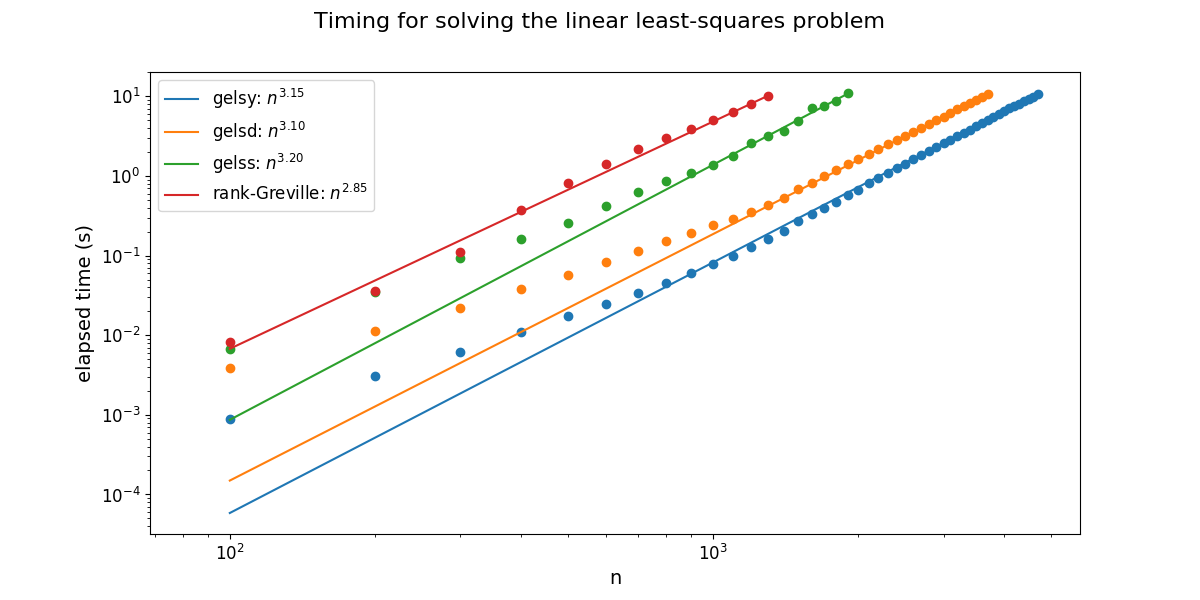}
\caption{\label{label_fig_cost_square}Timings for solving the linear least-squares problem on a random full-rank square observations matrix $R_\mathcal{N}(n,n,n)$. The asymptotic dependency with respect to $n$ is fitted on the last points and reported in the legend.}
\end{figure}

\begin{figure}[h!]
\centering
\includegraphics[width=\linewidth]{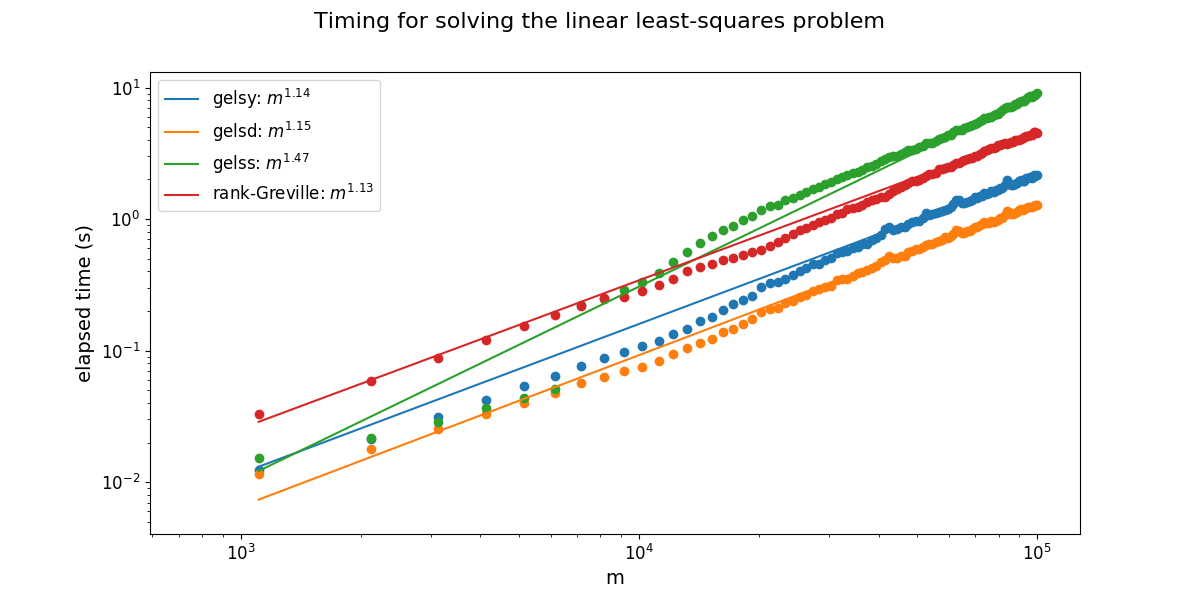}
\caption{\label{label_fig_cost_rect}Timings for solving the linear least-squares problem on a random full-rank observations matrix $R_\mathcal{N}(n,m,n)$ with a fixed number of observations $n=100$. The asymptotic dependency with respect to $m$ is fitted on the last points and reported in the legend.}
\end{figure}

\begin{figure}[h!]
\centering
\includegraphics[width=\linewidth]{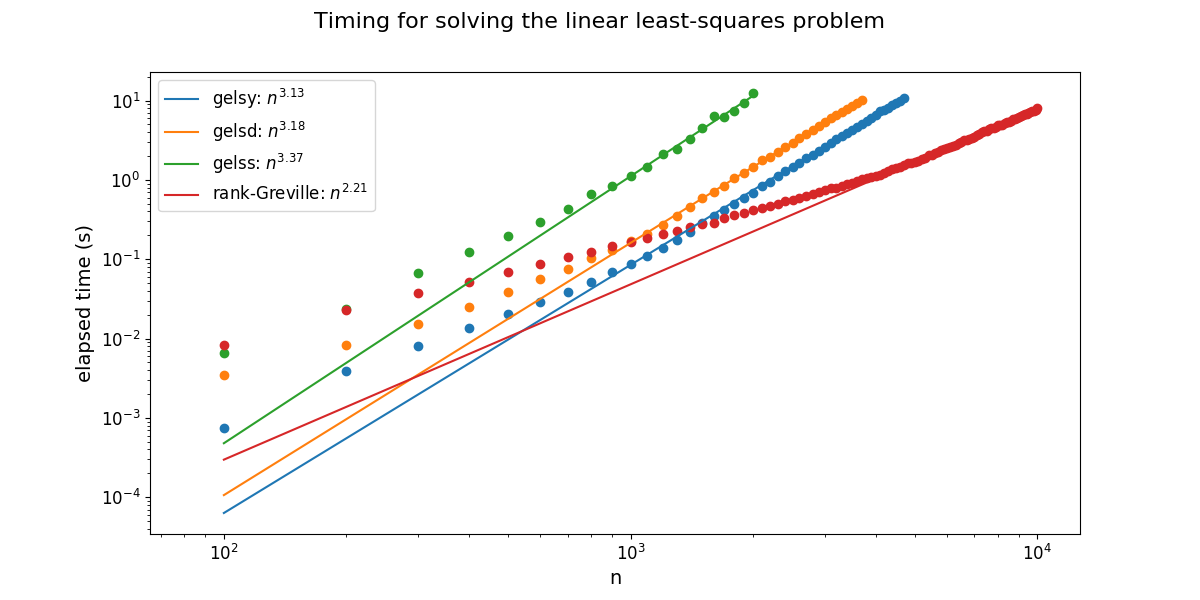}
\caption{\label{label_fig_cost_rank}Timings for solving the linear least-squares problem on a random rank-deficient square observations matrix $R_\mathcal{N}(n,n,r)$ with a fixed rank $r=100$. The asymptotic dependency with respect to $n$ is fitted on the last points and reported in the legend.}
\end{figure}

\begin{figure}[h!]
\centering
\includegraphics[width=\linewidth]{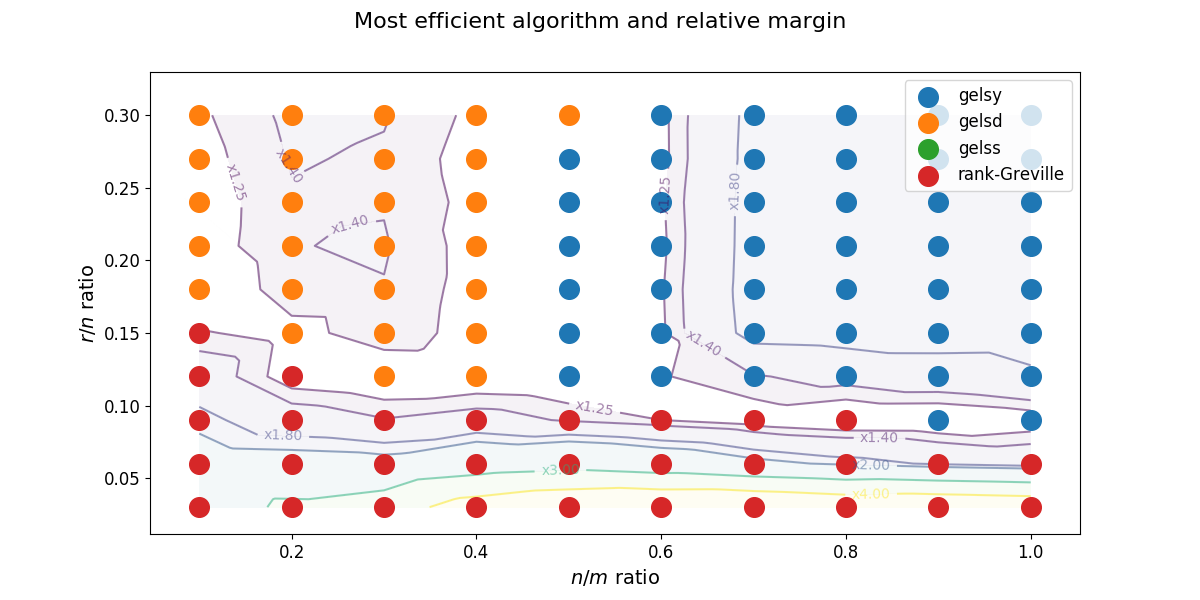}
\caption{\label{label_fig_domain_ratio}The fastest algorithm is represented for various $n/m$ and $r/n$ ratios, with $m = 4000$. The contour plot represents the interpolated relative margin by which an algorithm is the fastest (e.g. $\times1.25$ means that the execution time for \add{the} second fastest algorithm was $1.25$ times larger than for the fastest one).}
\end{figure}

These tests confirm the lowest $O(mnr)$ asymptotic complexity of the rank-Greville algorithm compared with other LAPACK solvers ($O(m^2n)$ or $O(mn^2)$) for solving the least-squares problem from scratch. We note, nonetheless, that rank-Greville has a larger pre-factor, in part due to the additional work of maintaining a constantly up-to-date minimum-norm least-squares solution, typical of RLS solvers. 
\clearpage

\subsection{Numerical stability}

In this section we evaluate the numerical stability achieved by the approach described in this paper for computing the pseudoinverse.

For a more meaningful analysis, we compare our rank-Greville algorithm with standard algorithms from the LAPACK library described in \ref{label_LAPACK_desc}, and other related algorithms:
\begin{itemize}
\item "Cholesky" refers to the \textit{scipy.linalg.cho\_solve} solver applied to the Cholesky factorization of $M = A^\top A$.
\item "Greville" refers to a basic implementation of the original Greville algorithm\cite{greville, ben_greville_generalized}.
\item "orthogonal" refers to the orthogonal variant described in section \ref{label_orthog_variant} with rescaling factors $\alpha_{n+1} = 1$.
\item "orthonormal" refers to the orthonormal variant described in section \ref{label_orthog_variant} (i.e.\ orthogonal variant with rescaling factors $\alpha_{n+1} = ||\Gnplusonerej||_2^{-1}$).
\end{itemize}

In order to assess the numerical stability of the algorithms described in this paper, we rely on the measures defined in \cite{numerical_tests, higham_accuracy_algo}:
\begin{itemize}
\item The stability factor of an algorithm with respect to the computation of the pseudoinverse $A^+$ of $A$ is given by:
\begin{equation}\label{label_eq_stability_factor}
e_\text{algo} = \frac{\left|\left|A^+_\text{algo} - A^+\right|\right|_2}{\epsilon_M \left|\left|A^+\right|\right|_2 \kappa_2(A)}
\end{equation}
where $A^+_\text{algo}$ is the pseudoinverse of $A$ computed by the algorithm, $A^+$ is the exact pseudoinverse, $||\cdot||_2$ is the 2-norm{ (e.g. for a matrix $A$, $||A||_2 = \max(\sigma(A))$ is the largest singular value of $A$)}, $\kappa_2(A) = \frac{\max(\sigma(A))}{\min(\sigma(A))}$ is the condition number of $A$ and $\epsilon_M$ is the machine precision.

This estimator is related to the forward stability of such algorithm, and should be bounded by a small constant depending on $m$ and $n$.
\item Similarly, we refer to the residual error as:
\begin{equation}\label{label_eq_residual_error}
res_\text{algo} = \frac{\left|\left|A^+_\text{algo}A - I\right|\right|_2}{\left|\left|A\right|\right|_2 \left|\left|A^+_\text{algo}\right|\right|_2}
\end{equation}
where $I$ is the identity{\iffalse{ on the appropriate ambient space}\fi}.

This estimator is related to the mixed forward-backward stability of such algorithm, and should be of the order of the machine precision $\epsilon_M$.
\end{itemize}

During the test, the machine precision used corresponds to $\epsilon_M \approx 2.22 \times 10^{-16}$.

This evaluation was performed empirically, using the matrices defined in \cite{numerical_tests}:
\begin{itemize}
\item Pascal matrices $P(n) \in \mathbb{Z}^{n\times n}$.

These are full-rank square matrices whose inverses $P(n)^{-1} \in \mathbb{Z}^{n\times n}$ can be computed exactly since their elements are also integers.

Empirical results are reported in Tables \ref{table_e_Pascal} and \ref{table_res_Pascal}. We found Greville-like algorithms, as well as the Cholesky decomposition to be orders of magnitudes less robust than the standard LAPACK solvers, with respect to both numerical stability indicators. Nonetheless, one should note that, as expected, when using \textit{fractions.Fraction} based numerical representation, the rank Greville algorithm and its orthogonal variant were able to compute the exact inverse.

\item Random matrices $R_\mathcal{N} \in \mathbb{R}^{3n\times n}$ whose elements are sampled from the normal $\mathcal{N}(0,1)$ distribution.

Pseudoinverses generated by the DGELSY solver were used as reference for computing the stability factor since they display the lowest residual error and the empirical results are reported in Tables \ref{table_e_random} and \ref{table_res_random}. We also found that the numerical stability indicators for the Greville-like algorithms are significantly dependent on the seed used by the pseudorandom number generator, unlike other algorithms tested.

\item Random ill-conditioned matrices ${R_\mathcal{N}}^4 \in \mathbb{R}^{n\times n}$, taken as the fourth power of random square matrices $R_\mathcal{N} \in \mathbb{R}^{n\times n}$ whose elements are sampled from the $\mathcal{N}(0,1)$ distribution.

Similarly as above, we used for pseudoinverse reference $\left({R_\mathcal{N}}^4\right)^+ = \left({R_\mathcal{N}}^+\right)^4$, the fourth power of the pseudoinverse generated by the DGELSY solver.

Empirical results are reported in Tables \ref{table_e_ill-cond} and \ref{table_res_ill-cond}, which show that for these ill-conditioned matrices, the Cholesky-based solver is, overall, the least stable algorithm tested herein.

\item Random matrices $USV^\top$, where $U \in \mathbb{R}^{5n\times n}$ and $V \in \mathbb{R}^{n\times n}$ are random column orthogonal matrices and {$S = diag(1, 2^\frac{1}{2}, \ldots, 2^\frac{n-1}{2})$}.

In this case, $(USV^\top)^+ = VS^{-1}U^\top$ was used for pseudoinverse reference.

Empirical results are reported in Tables \ref{table_e_USV} and \ref{table_res_USV}. For these tests, the \textit{rcond}/\textit{eps} parameter was set to $10^{-8}$. This was required by Greville-like algorithms to reach a reasonable solution.

\item Kahan matrices\cite{Kahan_1966} $K(c, s) \in \mathbb{R}^{n\times n}$ with $c^2 + s^2 = 1$ and $n=100$.

An explicit formula is available for the inverse\cite{Kahan_1966}, and was used as pseudoinverse reference.

Empirical results are reported in Table \ref{table_res_Kahan}. Unlike what was reported in \cite{numerical_tests}, we did not find the pure SVD-based solver to perform significantly worse than other LAPACK solvers. Furthermore, after setting the \textit{rcond}/\textit{eps} parameter low enough to tackle the extreme ill-conditionality of the Kahan matrices (i.e. $rcond < \kappa_2(A)^{-1}$), all algorithms (except Cholesky) performed relatively well, with the QR-based LAPACK solver performing the best.
\end{itemize}

\begin{table}\
		                  \caption{\label{table_e_Pascal}Empiric stability factors associated with the pseudoinverse computation of Pascal matrices $P(n)$.}\
		                  \makebox[\linewidth]{
\begin{tabular}{llllllllll}
\hline
 $n$   & $\kappa_2(A)$        & $e_\text{orthonormal}$   & $e_\text{orthogonal}$   & $e_\text{rank-Greville}$   & $e_\text{Greville}$   & $e_\text{Cholesky}$   & $e_\text{gelsy}$     & $e_\text{gelsd}$     & $e_\text{gelss}$     \\
\hline
 $4$   & $6.92\mathrm{e}{+2}$ & $8.80\mathrm{e}{-2}$     & $1.11\mathrm{e}{-1}$    & $1.67\mathrm{e}{+0}$       & $4.73\mathrm{e}{-2}$  & $5.88\mathrm{e}{+0}$  & $3.86\mathrm{e}{-2}$ & $6.67\mathrm{e}{-2}$ & $6.76\mathrm{e}{-2}$ \\
 $6$   & $1.11\mathrm{e}{+5}$ & $7.18\mathrm{e}{+2}$     & $1.06\mathrm{e}{+2}$    & $2.12\mathrm{e}{+2}$       & $6.58\mathrm{e}{+2}$  & $2.44\mathrm{e}{+3}$  & $5.04\mathrm{e}{-3}$ & $4.74\mathrm{e}{-2}$ & $4.74\mathrm{e}{-2}$ \\
 $8$   & $2.06\mathrm{e}{+7}$ & $3.42\mathrm{e}{+5}$     & $1.93\mathrm{e}{+3}$    & $2.18\mathrm{e}{+4}$       & $2.08\mathrm{e}{+5}$  & $1.59\mathrm{e}{+5}$  & $3.86\mathrm{e}{-3}$ & $8.55\mathrm{e}{-3}$ & $8.55\mathrm{e}{-3}$ \\
 $10$  & $4.16\mathrm{e}{+9}$ & $1.08\mathrm{e}{+6}$     & $1.08\mathrm{e}{+6}$    & $1.08\mathrm{e}{+6}$       & $1.08\mathrm{e}{+6}$  & $7.34\mathrm{e}{+5}$  & $3.00\mathrm{e}{-4}$ & $1.86\mathrm{e}{-3}$ & $1.86\mathrm{e}{-3}$ \\
\hline
\end{tabular}
}\end{table}
\begin{table}\
		                  \caption{\label{table_res_Pascal}Empiric residual errors associated with the pseudoinverse computation of Pascal matrices $P(n)$.}\
		                  \makebox[\linewidth]{
\begin{tabular}{llllllllll}
\hline
 $n$   & $\kappa_2(A)$        & $res_\text{orthonormal}$   & $res_\text{orthogonal}$   & $res_\text{rank-Greville}$   & $res_\text{Greville}$   & $res_\text{Cholesky}$   & $res_\text{gelsy}$    & $res_\text{gelsd}$    & $res_\text{gelss}$    \\
\hline
 $4$   & $6.92\mathrm{e}{+2}$ & $1.83\mathrm{e}{-15}$      & $5.99\mathrm{e}{-16}$     & $3.85\mathrm{e}{-16}$        & $1.66\mathrm{e}{-17}$   & $4.01\mathrm{e}{-15}$   & $5.29\mathrm{e}{-17}$ & $5.39\mathrm{e}{-17}$ & $5.42\mathrm{e}{-17}$ \\
 $6$   & $1.11\mathrm{e}{+5}$ & $1.81\mathrm{e}{-13}$      & $4.16\mathrm{e}{-14}$     & $4.71\mathrm{e}{-14}$        & $1.46\mathrm{e}{-13}$   & $1.04\mathrm{e}{-12}$   & $2.49\mathrm{e}{-17}$ & $4.77\mathrm{e}{-17}$ & $4.47\mathrm{e}{-17}$ \\
 $8$   & $2.06\mathrm{e}{+7}$ & $7.60\mathrm{e}{-11}$      & $6.17\mathrm{e}{-13}$     & $4.84\mathrm{e}{-12}$        & $4.61\mathrm{e}{-11}$   & $9.24\mathrm{e}{-11}$   & $6.06\mathrm{e}{-17}$ & $3.16\mathrm{e}{-17}$ & $1.58\mathrm{e}{-17}$ \\
 $10$  & $4.16\mathrm{e}{+9}$ & $1.42\mathrm{e}{-9}$       & $1.61\mathrm{e}{-9}$      & $1.37\mathrm{e}{-9}$         & $1.54\mathrm{e}{-9}$    & $2.35\mathrm{e}{-8}$    & $4.78\mathrm{e}{-17}$ & $3.26\mathrm{e}{-17}$ & $1.54\mathrm{e}{-17}$ \\
\hline
\end{tabular}
}\end{table}

\begin{table}\
		                  \caption{\label{table_e_random}Empiric stability factors associated with the pseudoinverse computation of random matrices $R_\mathcal{N} \in \mathbb{R}^{3n\times n}$ with elements distributed from $\mathcal{N}(0,1)$.}\
		                  \makebox[\linewidth]{
\begin{tabular}{llllllllll}
\hline
 $n$   & $\kappa_2(A)$        & $e_\text{orthonormal}$   & $e_\text{orthogonal}$   & $e_\text{rank-Greville}$   & $e_\text{Greville}$   & $e_\text{Cholesky}$   & $e_\text{gelsd}$     & $e_\text{gelss}$     \\
\hline
 $4$   & $2.44\mathrm{e}{+0}$ & $1.06\mathrm{e}{+2}$     & $7.46\mathrm{e}{+1}$    & $2.02\mathrm{e}{+1}$       & $2.84\mathrm{e}{+0}$  & $7.57\mathrm{e}{-1}$  & $2.08\mathrm{e}{+0}$ & $1.67\mathrm{e}{+0}$ \\
 $6$   & $2.25\mathrm{e}{+0}$ & $4.00\mathrm{e}{+0}$     & $4.98\mathrm{e}{+0}$    & $2.56\mathrm{e}{+0}$       & $1.05\mathrm{e}{+0}$  & $9.20\mathrm{e}{-1}$  & $1.35\mathrm{e}{+0}$ & $1.86\mathrm{e}{+0}$ \\
 $8$   & $2.74\mathrm{e}{+0}$ & $3.98\mathrm{e}{+0}$     & $4.43\mathrm{e}{+0}$    & $2.86\mathrm{e}{+0}$       & $2.23\mathrm{e}{+0}$  & $1.39\mathrm{e}{+0}$  & $1.81\mathrm{e}{+0}$ & $2.63\mathrm{e}{+0}$ \\
 $10$  & $2.84\mathrm{e}{+0}$ & $1.14\mathrm{e}{+3}$     & $4.60\mathrm{e}{+2}$    & $3.18\mathrm{e}{+2}$       & $2.24\mathrm{e}{+1}$  & $1.54\mathrm{e}{+0}$  & $1.99\mathrm{e}{+0}$ & $2.21\mathrm{e}{+0}$ \\
\hline
\end{tabular}
}\end{table}
\begin{table}\
		                  \caption{\label{table_res_random}Empiric residual errors associated with the pseudoinverse computation of random matrices $R_\mathcal{N} \in \mathbb{R}^{3n\times n}$ with elements distributed from $\mathcal{N}(0,1)$.}\
		                  \makebox[\linewidth]{
\begin{tabular}{llllllllll}
\hline
 $n$   & $\kappa_2(A)$        & $res_\text{orthonormal}$   & $res_\text{orthogonal}$   & $res_\text{rank-Greville}$   & $res_\text{Greville}$   & $res_\text{Cholesky}$   & $res_\text{gelsy}$    & $res_\text{gelsd}$    & $res_\text{gelss}$    \\
\hline
 $4$   & $2.44\mathrm{e}{+0}$ & $3.45\mathrm{e}{-14}$      & $2.55\mathrm{e}{-14}$     & $6.36\mathrm{e}{-15}$        & $7.20\mathrm{e}{-16}$   & $1.25\mathrm{e}{-16}$   & $2.24\mathrm{e}{-16}$ & $4.38\mathrm{e}{-16}$ & $4.68\mathrm{e}{-16}$ \\
 $6$   & $2.25\mathrm{e}{+0}$ & $8.53\mathrm{e}{-16}$      & $1.69\mathrm{e}{-15}$     & $7.61\mathrm{e}{-16}$        & $2.09\mathrm{e}{-16}$   & $2.32\mathrm{e}{-16}$   & $2.24\mathrm{e}{-16}$ & $4.10\mathrm{e}{-16}$ & $6.77\mathrm{e}{-16}$ \\
 $8$   & $2.74\mathrm{e}{+0}$ & $1.24\mathrm{e}{-15}$      & $1.05\mathrm{e}{-15}$     & $3.88\mathrm{e}{-16}$        & $3.39\mathrm{e}{-16}$   & $2.77\mathrm{e}{-16}$   & $2.48\mathrm{e}{-16}$ & $4.96\mathrm{e}{-16}$ & $6.16\mathrm{e}{-16}$ \\
 $10$  & $2.84\mathrm{e}{+0}$ & $4.75\mathrm{e}{-13}$      & $1.73\mathrm{e}{-13}$     & $3.73\mathrm{e}{-14}$        & $3.22\mathrm{e}{-15}$   & $4.36\mathrm{e}{-16}$   & $2.64\mathrm{e}{-16}$ & $6.33\mathrm{e}{-16}$ & $8.02\mathrm{e}{-16}$ \\
\hline
\end{tabular}
}\end{table}

\begin{table}\
		                  \caption{\label{table_e_ill-cond}Empiric stability factors associated with the pseudoinverse computation of random ill-conditioned matrices $R_\mathcal{N}^4 \in \mathbb{R}^{n\times n}$.}\
		                  \makebox[\linewidth]{
\begin{tabular}{llllllllll}
\hline
 $n$   & $\kappa_2(A)$        & $e_\text{orthonormal}$   & $e_\text{orthogonal}$   & $e_\text{rank-Greville}$   & $e_\text{Greville}$   & $e_\text{Cholesky}$   & $e_\text{gelsy}$     & $e_\text{gelsd}$     & $e_\text{gelss}$     \\
\hline
 $6$   & $3.39\mathrm{e}{+2}$ & $5.52\mathrm{e}{+0}$     & $3.04\mathrm{e}{+0}$    & $1.39\mathrm{e}{+1}$       & $8.87\mathrm{e}{+0}$  & $3.72\mathrm{e}{+1}$  & $5.59\mathrm{e}{-2}$ & $6.86\mathrm{e}{-2}$ & $7.36\mathrm{e}{-2}$ \\
 $8$   & $1.08\mathrm{e}{+2}$ & $1.50\mathrm{e}{+0}$     & $3.07\mathrm{e}{+0}$    & $4.16\mathrm{e}{+0}$       & $6.85\mathrm{e}{+0}$  & $2.40\mathrm{e}{+0}$  & $1.66\mathrm{e}{-1}$ & $1.58\mathrm{e}{-1}$ & $1.72\mathrm{e}{-1}$ \\
 $10$  & $2.92\mathrm{e}{+7}$ & $3.53\mathrm{e}{+0}$     & $2.62\mathrm{e}{+0}$    & $8.25\mathrm{e}{+0}$       & $5.62\mathrm{e}{+0}$  & $6.60\mathrm{e}{+5}$  & $2.07\mathrm{e}{-1}$ & $1.49\mathrm{e}{-2}$ & $1.49\mathrm{e}{-2}$ \\
 $12$  & $5.78\mathrm{e}{+4}$ & $2.47\mathrm{e}{+1}$     & $1.64\mathrm{e}{+1}$    & $2.41\mathrm{e}{+2}$       & $3.06\mathrm{e}{+1}$  & $1.62\mathrm{e}{+3}$  & $2.80\mathrm{e}{-2}$ & $7.38\mathrm{e}{-2}$ & $7.38\mathrm{e}{-2}$ \\
\hline
\end{tabular}
}\end{table}
\begin{table}\
		                  \caption{\label{table_res_ill-cond}Empiric residual errors associated with the pseudoinverse computation of random ill-conditioned matrices $R_\mathcal{N}^4 \in \mathbb{R}^{n\times n}$.}\
		                  \makebox[\linewidth]{
\begin{tabular}{llllllllll}
\hline
 $n$   & $\kappa_2(A)$        & $res_\text{orthonormal}$   & $res_\text{orthogonal}$   & $res_\text{rank-Greville}$   & $res_\text{Greville}$   & $res_\text{Cholesky}$   & $res_\text{gelsy}$    & $res_\text{gelsd}$    & $res_\text{gelss}$    \\
\hline
 $6$   & $3.39\mathrm{e}{+2}$ & $1.40\mathrm{e}{-15}$      & $7.32\mathrm{e}{-16}$     & $3.38\mathrm{e}{-15}$        & $2.03\mathrm{e}{-15}$   & $1.01\mathrm{e}{-14}$   & $5.06\mathrm{e}{-17}$ & $8.23\mathrm{e}{-17}$ & $4.55\mathrm{e}{-17}$ \\
 $8$   & $1.08\mathrm{e}{+2}$ & $8.97\mathrm{e}{-16}$      & $1.55\mathrm{e}{-15}$     & $1.05\mathrm{e}{-15}$        & $1.63\mathrm{e}{-15}$   & $2.90\mathrm{e}{-15}$   & $1.02\mathrm{e}{-16}$ & $3.00\mathrm{e}{-16}$ & $2.79\mathrm{e}{-16}$ \\
 $10$  & $2.92\mathrm{e}{+7}$ & $1.57\mathrm{e}{-15}$      & $8.01\mathrm{e}{-16}$     & $1.83\mathrm{e}{-15}$        & $1.25\mathrm{e}{-15}$   & $9.58\mathrm{e}{-10}$   & $6.37\mathrm{e}{-17}$ & $1.71\mathrm{e}{-16}$ & $1.58\mathrm{e}{-16}$ \\
 $12$  & $5.78\mathrm{e}{+4}$ & $6.09\mathrm{e}{-15}$      & $5.01\mathrm{e}{-15}$     & $5.36\mathrm{e}{-14}$        & $6.90\mathrm{e}{-15}$   & $1.55\mathrm{e}{-12}$   & $4.71\mathrm{e}{-17}$ & $9.97\mathrm{e}{-17}$ & $1.36\mathrm{e}{-16}$ \\
\hline
\end{tabular}
}\end{table}

\begin{table}\
		                  \caption{\label{table_e_USV}Empiric stability factors associated with the pseudoinverse computation of random matrices $USV^\top \in \mathbb{R}^{5n\times n}$, where $U$ and $V$ are random column orthogonal matrices and $S = diag(1, 2^\frac{1}{2}, \ldots, 2^\frac{n-1}{2})$.}\
		                  \makebox[\linewidth]{
\begin{tabular}{llllllllll}
\hline
 $n$   & $\kappa_2(A)$        & $e_\text{orthonormal}$   & $e_\text{orthogonal}$   & $e_\text{rank-Greville}$   & $e_\text{Greville}$   & $e_\text{Cholesky}$   & $e_\text{gelsy}$     & $e_\text{gelsd}$     & $e_\text{gelss}$     \\
\hline
 $10$  & $2.26\mathrm{e}{+1}$ & $6.72\mathrm{e}{+0}$     & $4.60\mathrm{e}{+0}$    & $4.19\mathrm{e}{+1}$       & $1.09\mathrm{e}{+1}$  & $1.85\mathrm{e}{+0}$  & $2.77\mathrm{e}{-1}$ & $2.94\mathrm{e}{-1}$ & $2.89\mathrm{e}{-1}$ \\
 $15$  & $1.28\mathrm{e}{+2}$ & $1.63\mathrm{e}{+2}$     & $1.60\mathrm{e}{+2}$    & $2.40\mathrm{e}{+3}$       & $2.99\mathrm{e}{+3}$  & $1.16\mathrm{e}{+1}$  & $1.70\mathrm{e}{-1}$ & $2.85\mathrm{e}{-1}$ & $2.93\mathrm{e}{-1}$ \\
 $20$  & $7.24\mathrm{e}{+2}$ & $1.57\mathrm{e}{+2}$     & $2.41\mathrm{e}{+2}$    & $1.52\mathrm{e}{+4}$       & $1.07\mathrm{e}{+4}$  & $4.24\mathrm{e}{+1}$  & $1.60\mathrm{e}{-1}$ & $1.85\mathrm{e}{-1}$ & $1.85\mathrm{e}{-1}$ \\
\hline
\end{tabular}
}\end{table}
\begin{table}\
		                  \caption{\label{table_res_USV}Empiric residual errors associated with the pseudoinverse computation of random matrices $USV^\top \in \mathbb{R}^{5n\times n}$, where $U$ and $V$ are random column orthogonal matrices and $S = diag(1, 2^\frac{1}{2}, \ldots, 2^\frac{n-1}{2})$.}\
		                  \makebox[\linewidth]{
\begin{tabular}{llllllllll}
\hline
 $n$   & $\kappa_2(A)$        & $res_\text{orthonormal}$   & $res_\text{orthogonal}$   & $res_\text{rank-Greville}$   & $res_\text{Greville}$   & $res_\text{Cholesky}$   & $res_\text{gelsy}$    & $res_\text{gelsd}$    & $res_\text{gelss}$    \\
\hline
 $10$  & $2.26\mathrm{e}{+1}$ & $2.68\mathrm{e}{-15}$      & $4.25\mathrm{e}{-15}$     & $2.52\mathrm{e}{-15}$        & $5.92\mathrm{e}{-16}$   & $1.70\mathrm{e}{-15}$   & $1.01\mathrm{e}{-16}$ & $1.62\mathrm{e}{-16}$ & $1.26\mathrm{e}{-16}$ \\
 $15$  & $1.28\mathrm{e}{+2}$ & $1.67\mathrm{e}{-14}$      & $9.69\mathrm{e}{-15}$     & $9.20\mathrm{e}{-14}$        & $2.59\mathrm{e}{-13}$   & $7.37\mathrm{e}{-15}$   & $7.58\mathrm{e}{-17}$ & $1.19\mathrm{e}{-16}$ & $1.30\mathrm{e}{-16}$ \\
 $20$  & $7.24\mathrm{e}{+2}$ & $1.65\mathrm{e}{-14}$      & $6.11\mathrm{e}{-14}$     & $6.04\mathrm{e}{-13}$        & $4.93\mathrm{e}{-13}$   & $3.53\mathrm{e}{-14}$   & $4.31\mathrm{e}{-17}$ & $9.96\mathrm{e}{-17}$ & $5.61\mathrm{e}{-17}$ \\
\hline
\end{tabular}
}\end{table}

\begin{table}\
		                  \caption{\label{table_res_Kahan}Empiric residual errors associated with the pseudoinverse computation of Kahan matrices $K(c, s) \in \mathbb{R}^{100\times 100}$, with $c^2 + s^2 = 1$.}\
		                  \makebox[\linewidth]{
\begin{tabular}{llllllllll}
\hline
 c      & $\kappa_2(A)$         & $res_\text{orthonormal}$   & $res_\text{orthogonal}$   & $res_\text{rank-Greville}$   & $res_\text{Greville}$   & $res_\text{Cholesky}$   & $res_\text{gelsy}$    & $res_\text{gelsd}$    & $res_\text{gelss}$    \\
\hline
 $0.10$ & $5.42\mathrm{e}{+4}$  & $5.57\mathrm{e}{-17}$      & $3.00\mathrm{e}{-17}$     & $3.50\mathrm{e}{-17}$        & $2.64\mathrm{e}{-17}$   & $3.40\mathrm{e}{-13}$   & $5.24\mathrm{e}{-17}$ & $7.61\mathrm{e}{-17}$ & $1.33\mathrm{e}{-16}$ \\
 $0.15$ & $1.13\mathrm{e}{+7}$  & $1.11\mathrm{e}{-17}$      & $2.16\mathrm{e}{-17}$     & $1.03\mathrm{e}{-17}$        & $1.19\mathrm{e}{-17}$   & $1.60\mathrm{e}{-10}$   & $1.52\mathrm{e}{-17}$ & $5.61\mathrm{e}{-17}$ & $7.97\mathrm{e}{-17}$ \\
 $0.20$ & $2.18\mathrm{e}{+9}$  & $8.10\mathrm{e}{-18}$      & $7.96\mathrm{e}{-18}$     & $2.42\mathrm{e}{-18}$        & $2.41\mathrm{e}{-18}$   & failure                 & $7.40\mathrm{e}{-18}$ & $6.30\mathrm{e}{-17}$ & $5.69\mathrm{e}{-17}$ \\
 $0.25$ & $4.37\mathrm{e}{+11}$ & $2.07\mathrm{e}{-18}$      & $1.06\mathrm{e}{-18}$     & $1.14\mathrm{e}{-18}$        & $1.17\mathrm{e}{-18}$   & $1.61\mathrm{e}{-6}$    & $8.74\mathrm{e}{-18}$ & $4.85\mathrm{e}{-17}$ & $4.39\mathrm{e}{-17}$ \\
 $0.30$ & $9.77\mathrm{e}{+13}$ & $3.01\mathrm{e}{-19}$      & $4.31\mathrm{e}{-19}$     & $1.92\mathrm{e}{-19}$        & $1.94\mathrm{e}{-19}$   & $7.98\mathrm{e}{-4}$    & $1.82\mathrm{e}{-19}$ & $1.33\mathrm{e}{-16}$ & $2.43\mathrm{e}{-17}$ \\
 $0.35$ & $2.57\mathrm{e}{+16}$ & $3.57\mathrm{e}{-20}$      & $4.27\mathrm{e}{-20}$     & $2.73\mathrm{e}{-20}$        & $2.18\mathrm{e}{-20}$   & failure                 & $5.05\mathrm{e}{-20}$ & $9.43\mathrm{e}{-18}$ & $2.03\mathrm{e}{-18}$ \\
 $0.40$ & $8.36\mathrm{e}{+18}$ & $3.46\mathrm{e}{-21}$      & $3.51\mathrm{e}{-21}$     & $3.09\mathrm{e}{-21}$        & $2.55\mathrm{e}{-21}$   & $1.02\mathrm{e}{-4}$    & $1.08\mathrm{e}{-20}$ & $3.17\mathrm{e}{-19}$ & $3.04\mathrm{e}{-19}$ \\
\hline
\end{tabular}
}\end{table}

The algorithms described in this paper (rank-Greville and variants), including the original Greville algorithm, perform roughly equivalently in terms of numerical stability. Furthermore, the stability of these Greville-like algorithms seems much less dependent on $\kappa_2(A)$ than the Cholesky based algorithm. As expected, we found these algorithms to be neither mixed forward-backward nor forward stable. As a consequence, the much more robust LAPACK least-squares solvers are to be recommended when numerical stability is crucial. However, the stability of Greville-like algorithms are competitive compared to the Cholesky-based LS solvers.

A compromise between update efficiency and numerical stability can be searched among the full QR decomposition updating algorithms\cite{bjorck1996numerical}, or even faster, stable updating algorithms for Gram-Schmidt QR factorization\cite{daniel1976reorthogonalization, bjorck1996numerical}. 
\clearpage
\section{Conclusion}

In this paper, we first derive a simple explicit formula for the recursive least-squares problem using a general rank decomposition update scheme.
Based on this, we devise a transparent, Greville-like algorithm. We also introduce two variants bridging the gap between Greville's algorithm and QR-based least-squares solvers.
In contrast to Greville's algorithm, we maintain a rank decomposition at each update step. This allows us to exploit rank-deficiency to reach an asymptotic computational complexity of $O(mr)$ for updating a least-squares solution when adding an observation, leading to a total complexity of $O(mnr)$ for computing the full least-squares solution. 
This complexity is lower than Greville's algorithm or any commonly available solver tested, even though a truncated QR factorization based solver can achieve such a $O(mnr)$ bound for computing the full least-squares solution\cite{qr_truncated}. Nonetheless, a $O(mr)$ bound for the least-squares solution update step is, to our knowledge, lower than those achieved by the more sophisticated updating algorithms explicitly reported in the literature. 
We have implemented these algorithms in Python3, using Numpy. This publicly available implementation offers a recursive least-squares solver ($O(mr)$ operations per update), with optional pseudoinverse ($O(mn)$) and covariance support ($O(m^2)$). 
The numerical stability of these Greville-like algorithms were empirically found to be significantly inferior compared to common LAPACK solvers. However, it is important to note that the algebraic simplicity of some of these Greville-like methods make them compatible with exact numerical representation, without the need to use symbolic computing.

\bibliography{biblio}

\begin{thebibliography}{25}
\providecommand{\natexlab}[1]{#1}
\providecommand{\url}[1]{\texttt{#1}}
\providecommand{\href}[2]{#2}
\providecommand{\path}[1]{#1}
\providecommand{\eprint}[1]{\href{http://arxiv.org/abs/#1}{\path{#1}}}
\providecommand{\DOIprefix}{doi:}
\providecommand{\ArXivprefix}{arXiv:}
\providecommand{\URLprefix}{URL: }
\providecommand{\Pubmedprefix}{pmid:}
\providecommand{\doi}[1]{\href{http://dx.doi.org/#1}{\path{#1}}}
\providecommand{\Pubmed}[1]{\href{pmid:#1}{\path{#1}}}
\providecommand{\BIBand}{and}
\providecommand{\bibinfo}[2]{#2}
\ifx\xfnm\undefined \def\xfnm[#1]{\unskip,\space#1}\fi
\bibitem[{Bjorck(1996)}]{bjorck1996numerical}
\bibinfo{author}{Bjorck\xfnm[ A.]}.
\newblock \bibinfo{title}{Numerical Methods for Least Squares Problems}.
\newblock Other Titles in Applied Mathematics; \bibinfo{publisher}{Society for
  Industrial and Applied Mathematics}; \bibinfo{year}{1996}.
\newblock ISBN \bibinfo{isbn}{9780898713602}.
\newblock \URLprefix \url{https://doi.org/10.1137/1.9781611971484}.
\bibitem[{Penrose(1956)}]{penrose_1956_least_squares}
\bibinfo{author}{Penrose\xfnm[ R.]}.
\newblock \bibinfo{title}{On best approximate solutions of linear matrix
  equations}.
\newblock \bibinfo{journal}{Mathematical Proceedings of the Cambridge
  Philosophical Society}
  \bibinfo{year}{1956};\bibinfo{volume}{52}(\bibinfo{number}{1}):\bibinfo{pages}{17–19}.
\newblock \DOIprefix\doi{10.1017/S0305004100030929}.
\bibitem[{Rakha(2004)}]{rakha_moorepenrose_2004}
\bibinfo{author}{Rakha\xfnm[ M.A.]}.
\newblock \bibinfo{title}{On the {Moore}{\textendash}{Penrose} generalized
  inverse matrix}.
\newblock \bibinfo{journal}{Applied Mathematics and Computation}
  \bibinfo{year}{2004};\bibinfo{volume}{158}(\bibinfo{number}{1}):\bibinfo{pages}{185--200}.
\newblock \URLprefix
  \url{https://linkinghub.elsevier.com/retrieve/pii/S0096300303009998}.
  \DOIprefix\doi{10.1016/j.amc.2003.09.004}.
\bibitem[{Toutounian and Ataei(2009)}]{toutounian_new_2009}
\bibinfo{author}{Toutounian\xfnm[ F.]}, \bibinfo{author}{Ataei\xfnm[ A.]}.
\newblock \bibinfo{title}{A new method for computing
  {Moore}{\textendash}{Penrose} inverse matrices}.
\newblock \bibinfo{journal}{Journal of Computational and Applied Mathematics}
  \bibinfo{year}{2009};\bibinfo{volume}{228}(\bibinfo{number}{1}):\bibinfo{pages}{412--417}.
\newblock \URLprefix
  \url{https://linkinghub.elsevier.com/retrieve/pii/S0377042708005062}.
  \DOIprefix\doi{10.1016/j.cam.2008.10.008}.
\bibitem[{Moore(1920)}]{moore1920reciprocalmatrix}
\bibinfo{author}{Moore\xfnm[ E.]}.
\newblock \bibinfo{title}{On the reciprocal of the general algebraic matrix}.
\newblock \bibinfo{journal}{Bull Amer Math Soc}
  \bibinfo{year}{1920};\bibinfo{volume}{26}(\bibinfo{number}{9}):\bibinfo{pages}{394--395}.
\newblock \URLprefix \url{https://doi.org/10.1090/S0002-9904-1920-03322-7}.
  \DOIprefix\doi{10.1090/S0002-9904-1920-03322-7}.
\bibitem[{Penrose(1955)}]{penrose_1955_pseudoinverse}
\bibinfo{author}{Penrose\xfnm[ R.]}.
\newblock \bibinfo{title}{A generalized inverse for matrices}.
\newblock \bibinfo{journal}{Mathematical Proceedings of the Cambridge
  Philosophical Society}
  \bibinfo{year}{1955};\bibinfo{volume}{51}(\bibinfo{number}{3}):\bibinfo{pages}{406–413}.
\newblock \DOIprefix\doi{10.1017/S0305004100030401}.
\bibitem[{Bj{\"o}rck(2009)}]{least_squares_overview}
\bibinfo{author}{Bj{\"o}rck\xfnm[ {\AA}.]}.
\newblock \bibinfo{title}{Least Squares Problems}.
\newblock \bibinfo{address}{Boston, MA}: \bibinfo{publisher}{Springer US}.
\newblock ISBN \bibinfo{isbn}{978-0-387-74759-0}; \bibinfo{year}{2009}, p.
  \bibinfo{pages}{1856--1866}.
\newblock \URLprefix \url{https://doi.org/10.1007/978-0-387-74759-0_329}.
  \DOIprefix\doi{10.1007/978-0-387-74759-0_329}.
\bibitem[{Courrieu(2005)}]{Courrieu2005FastCO}
\bibinfo{author}{Courrieu\xfnm[ P.]}.
\newblock \bibinfo{title}{Fast computation of moore-penrose inverse matrices}.
\newblock \bibinfo{journal}{ArXiv}
  \bibinfo{year}{2005};\bibinfo{volume}{abs/0804.4809}.
\bibitem[{Greville(1960)}]{greville}
\bibinfo{author}{Greville\xfnm[ T.N.E.]}.
\newblock \bibinfo{title}{Some applications of the pseudoinverse of a matrix}.
\newblock \bibinfo{journal}{SIAM Review}
  \bibinfo{year}{1960};\bibinfo{volume}{2}(\bibinfo{number}{1}):\bibinfo{pages}{15--22}.
\newblock \URLprefix \url{https://doi.org/10.1137/1002004}.
  \DOIprefix\doi{10.1137/1002004}.
  \href{http://arxiv.org/abs/https://doi.org/10.1137/1002004}{\tt
  arXiv:https://doi.org/10.1137/1002004}.
\bibitem[{Albert and Sittler(1965)}]{albert_sittler}
\bibinfo{author}{Albert\xfnm[ A.]}, \bibinfo{author}{Sittler\xfnm[ R.W.]}.
\newblock \bibinfo{title}{A method for computing least squares estimators that
  keep up with the data}.
\newblock \bibinfo{journal}{Journal of the Society for Industrial and Applied
  Mathematics Series A Control}
  \bibinfo{year}{1965};\bibinfo{volume}{3}(\bibinfo{number}{3}):\bibinfo{pages}{384--417}.
\newblock \URLprefix \url{https://doi.org/10.1137/0303026}.
  \DOIprefix\doi{10.1137/0303026}.
  \href{http://arxiv.org/abs/https://doi.org/10.1137/0303026}{\tt
  arXiv:https://doi.org/10.1137/0303026}.
\bibitem[{{Le Gall}(2012)}]{fast_rect_dot}
\bibinfo{author}{{Le Gall}\xfnm[ F.]}.
\newblock \bibinfo{title}{Faster algorithms for rectangular matrix
  multiplication}.
\newblock In: \bibinfo{booktitle}{2012 IEEE 53rd Annual Symposium on
  Foundations of Computer Science}. \bibinfo{year}{2012}, p.
  \bibinfo{pages}{514--523}.
\newblock \DOIprefix\doi{10.1109/FOCS.2012.80}.
\bibitem[{Mirsky(1990)}]{mirsky1990introduction}
\bibinfo{author}{Mirsky\xfnm[ L.]}.
\newblock \bibinfo{title}{An Introduction to Linear Algebra}.
\newblock Dover Books on Mathematics; \bibinfo{publisher}{Dover};
  \bibinfo{year}{1990}.
\newblock ISBN \bibinfo{isbn}{9780486664347}.
\newblock \URLprefix \url{https://books.google.fr/books?id=ULMmheb26ZcC}.
\bibitem[{Greene(2003)}]{greene2003econometric}
\bibinfo{author}{Greene\xfnm[ W.]}.
\newblock \bibinfo{title}{Econometric Analysis}.
\newblock \bibinfo{publisher}{Prentice Hall}; \bibinfo{year}{2003}.
\newblock ISBN \bibinfo{isbn}{9780130661890}.
\bibitem[{Van~Rossum and Drake(2009)}]{Python3}
\bibinfo{author}{Van~Rossum\xfnm[ G.]}, \bibinfo{author}{Drake\xfnm[ F.L.]}.
\newblock \bibinfo{title}{Python 3 Reference Manual}.
\newblock \bibinfo{address}{Paramount, CA}: \bibinfo{publisher}{CreateSpace};
  \bibinfo{year}{2009}.
\newblock ISBN \bibinfo{isbn}{1441412697, 9781441412690}.
\bibitem[{P\'erez and Granger(2007)}]{IPython}
\bibinfo{author}{P\'erez\xfnm[ F.]}, \bibinfo{author}{Granger\xfnm[ B.E.]}.
\newblock \bibinfo{title}{{IP}ython: a system for interactive scientific
  computing}.
\newblock \bibinfo{journal}{Computing in Science and Engineering}
  \bibinfo{year}{2007};\bibinfo{volume}{9}(\bibinfo{number}{3}):\bibinfo{pages}{21--29}.
\newblock \URLprefix \url{https://ipython.org}.
  \DOIprefix\doi{10.1109/MCSE.2007.53}.
\bibitem[{{Virtanen} et~al.(2020){Virtanen}, {Gommers}, {Oliphant},
  {Haberland}, {Reddy}, {Cournapeau} et~al.}]{2020SciPy-NMeth}
\bibinfo{author}{{Virtanen}\xfnm[ P.]}, \bibinfo{author}{{Gommers}\xfnm[ R.]},
  \bibinfo{author}{{Oliphant}\xfnm[ T.E.]}, \bibinfo{author}{{Haberland}\xfnm[
  M.]}, \bibinfo{author}{{Reddy}\xfnm[ T.]},
  \bibinfo{author}{{Cournapeau}\xfnm[ D.]}, et~al.
\newblock \bibinfo{title}{{SciPy 1.0: Fundamental Algorithms for Scientific
  Computing in Python}}.
\newblock \bibinfo{journal}{Nature Methods}
  \bibinfo{year}{2020};\bibinfo{volume}{17}:\bibinfo{pages}{261--272}.
\newblock \DOIprefix\doi{https://doi.org/10.1038/s41592-019-0686-2}.
\bibitem[{Oliphant(2015)}]{Numpy}
\bibinfo{author}{Oliphant\xfnm[ T.E.]}.
\newblock \bibinfo{title}{Guide to NumPy}.
\newblock \bibinfo{edition}{2nd} ed.; \bibinfo{address}{USA}:
  \bibinfo{publisher}{CreateSpace Independent Publishing Platform};
  \bibinfo{year}{2015}.
\newblock ISBN \bibinfo{isbn}{151730007X, 9781517300074}.
\bibitem[{Walt et~al.(2011)Walt, Colbert and Varoquaux}]{Numpy2}
\bibinfo{author}{Walt\xfnm[ S.v.d.]}, \bibinfo{author}{Colbert\xfnm[ S.C.]},
  \bibinfo{author}{Varoquaux\xfnm[ G.]}.
\newblock \bibinfo{title}{The numpy array: A structure for efficient numerical
  computation}.
\newblock \bibinfo{journal}{Computing in Science \& Engineering}
  \bibinfo{year}{2011};\bibinfo{volume}{13}(\bibinfo{number}{2}):\bibinfo{pages}{22--30}.
\newblock \URLprefix
  \url{https://aip.scitation.org/doi/abs/10.1109/MCSE.2011.37}.
  \DOIprefix\doi{10.1109/MCSE.2011.37}.
\bibitem[{Anderson et~al.(1999)Anderson, Bai, Bischof, Blackford, Demmel,
  Dongarra et~al.}]{lapack_guide}
\bibinfo{author}{Anderson\xfnm[ E.]}, \bibinfo{author}{Bai\xfnm[ Z.]},
  \bibinfo{author}{Bischof\xfnm[ C.]}, \bibinfo{author}{Blackford\xfnm[ L.S.]},
  \bibinfo{author}{Demmel\xfnm[ J.]}, \bibinfo{author}{Dongarra\xfnm[ J.]},
  et~al.
\newblock \bibinfo{title}{LAPACK Users' Guide}.
\newblock \bibinfo{edition}{Third} ed.; \bibinfo{publisher}{Society for
  Industrial and Applied Mathematics}; \bibinfo{year}{1999}.
\newblock \URLprefix
  \url{https://epubs.siam.org/doi/abs/10.1137/1.9780898719604}.
  \DOIprefix\doi{10.1137/1.9780898719604}.
  \href{http://arxiv.org/abs/https://epubs.siam.org/doi/pdf/10.1137/1.9780898719604}{\tt
  arXiv:https://epubs.siam.org/doi/pdf/10.1137/1.9780898719604}.
\bibitem[{Higham(2002)}]{higham_accuracy_algo}
\bibinfo{author}{Higham\xfnm[ N.J.]}.
\newblock \bibinfo{title}{Accuracy and Stability of Numerical Algorithms}.
\newblock \bibinfo{edition}{Second} ed.; \bibinfo{publisher}{Society for
  Industrial and Applied Mathematics}; \bibinfo{year}{2002}.
\newblock \URLprefix
  \url{https://epubs.siam.org/doi/abs/10.1137/1.9780898718027}.
  \DOIprefix\doi{10.1137/1.9780898718027}.
  \href{http://arxiv.org/abs/https://epubs.siam.org/doi/pdf/10.1137/1.9780898718027}{\tt
  arXiv:https://epubs.siam.org/doi/pdf/10.1137/1.9780898718027}.
\bibitem[{Ben-Israel and Greville(2003)}]{ben_greville_generalized}
\bibinfo{author}{Ben-Israel\xfnm[ A.]}, \bibinfo{author}{Greville\xfnm[ T.N.]}.
\newblock \bibinfo{title}{Generalized inverses: theory and applications};
  vol.~\bibinfo{volume}{15}.
\newblock \bibinfo{publisher}{Springer Science \& Business Media};
  \bibinfo{year}{2003}.
\bibitem[{Smoktunowicz and Wrobel(2012)}]{numerical_tests}
\bibinfo{author}{Smoktunowicz\xfnm[ A.]}, \bibinfo{author}{Wrobel\xfnm[ I.]}.
\newblock \bibinfo{title}{Numerical aspects of computing the moore-penrose
  inverse of full column rank matrices}.
\newblock \bibinfo{journal}{BIT Numerical Mathematics}
  \bibinfo{year}{2012};\bibinfo{volume}{52}(\bibinfo{number}{2}):\bibinfo{pages}{503--524}.
\bibitem[{Kahan(1966)}]{Kahan_1966}
\bibinfo{author}{Kahan\xfnm[ W.]}.
\newblock \bibinfo{title}{Numerical linear algebra}.
\newblock \bibinfo{journal}{Canadian Mathematical Bulletin}
  \bibinfo{year}{1966};\bibinfo{volume}{9}(\bibinfo{number}{5}):\bibinfo{pages}{757–801}.
\newblock \DOIprefix\doi{10.4153/CMB-1966-083-2}.
\bibitem[{Daniel et~al.(1976)Daniel, Gragg, Kaufman and
  Stewart}]{daniel1976reorthogonalization}
\bibinfo{author}{Daniel\xfnm[ J.W.]}, \bibinfo{author}{Gragg\xfnm[ W.B.]},
  \bibinfo{author}{Kaufman\xfnm[ L.]}, \bibinfo{author}{Stewart\xfnm[ G.W.]}.
\newblock \bibinfo{title}{Reorthogonalization and stable algorithms for
  updating the gram-schmidt qr factorization}.
\newblock \bibinfo{journal}{Mathematics of Computation}
  \bibinfo{year}{1976};\bibinfo{volume}{30}(\bibinfo{number}{136}):\bibinfo{pages}{772--795}.
\bibitem[{Businger and Golub(1965)}]{qr_truncated}
\bibinfo{author}{Businger\xfnm[ P.]}, \bibinfo{author}{Golub\xfnm[ G.H.]}.
\newblock \bibinfo{title}{Linear least squares solutions by householder
  transformations}.
\newblock \bibinfo{journal}{Numer Math}
  \bibinfo{year}{1965};\bibinfo{volume}{7}(\bibinfo{number}{3}):\bibinfo{pages}{269–276}.
\newblock \URLprefix \url{https://doi.org/10.1007/BF01436084}.
  \DOIprefix\doi{10.1007/BF01436084}.

\end{thebibliography}
\end{document}